\documentclass[english,final]{IEEEtran}

\usepackage{graphicx}
\usepackage{amssymb}
\usepackage{amstext}
\usepackage{algpseudocode}
\usepackage{amscd}
\usepackage{algorithm}

\usepackage{verbatim}

\usepackage{tikz}
\usetikzlibrary{chains}
\usetikzlibrary{fit}
\usetikzlibrary{arrows}
\usetikzlibrary{decorations}
\usepackage{xcolor}
\usepackage{epsfig}
\usetikzlibrary{shapes.symbols,patterns} 
\usepackage{pgfplots}

\usepackage{amsthm}
\usepackage{amsmath}
\usepackage{amsfonts}
\usepackage{amssymb,amstext}
\usepackage{bbm} 
\usepackage{dsfont}

\usepackage[colorlinks=true,urlcolor=blue, hyperindex,breaklinks=true] {hyperref}
\usepackage{todonotes}

\theoremstyle{plain}
\newtheorem{mythm}{Theorem}
\newtheorem{myprop}[mythm]{Proposition}
\newtheorem{mycor}[mythm]{Corollary}
\newtheorem{mylem}[mythm]{Lemma}

\theoremstyle{definition}
\newtheorem{mydef}{Definition}

\newtheorem{myquestion}{Open Question}

\newcommand{\prlsection}[1]{\emph{#1}.---}
\newcommand{\ket}[1]{\left| #1 \right \rangle}
\newcommand{\bra}[1]{\left \langle #1 \right|}
\newcommand{\ketbra}[1]{\ket{#1}\bra{#1}}

\newcommand{\norm}[1]{\left\lVert#1\right\rVert}

\def\Hb{\ensuremath{H_{\rm b}}}
\def\S{\mathsf{S}}
\def\D{\mathsf{D}}
\def\C{\mathsf{C}}

\def\E{\mathsf{E}}
\def\F{\mathsf{F}}
\def\R{\mathsf{R}}

\def\setC{\mathsf{c}}

\def\TPCP{\mathrm{TPCP}}

\DeclareMathOperator{\dec}{dec}
\DeclareMathOperator{\enc}{enc}
\newcommand{\Hh}[1]{H\!\left({#1}\right)} 
\newcommand{\Hc}[2]{H\!\left({#1}\!\left|{#2}\right.\right)} 
\newcommand{\CI}[2]{I\!\left({#1}\rangle{#2}\right)} 
\newcommand{\I}[2]{I\!\left({#1}\!:\!{#2}\right)} 

\newcommand{\Zc}[2]{Z\!\left({#1}\!\left|{#2}\right.\right)} 
\newcommand{\Fi}[2]{F\!\left({#1},{#2}\right)} 

\newcommand{\W}{\mathsf{W}} 
\newcommand{\Prvcond}[2]{\,{\rm Pr}\!\left[\left.#1\,\right|#2\right]} 
\newcommand{\Prv}[1]{\,{\rm Pr}\!\left[#1\right]} 

\newcommand{\Bernoulli}[1]{\textnormal{Bernoulli}\left(#1\right)}  

\newcommand{\Tr}[1]{{\rm Tr}\!\left[{#1}\right]} 
\newcommand{\Trp}[2]{{\rm Tr}_{{#1}}\!\left({#2}\right)} 
\DeclareMathOperator{\Trs}{Tr}
\newcommand{\TD}[2]{\delta\!\left({#1},{#2}\right)} 

\newcommand{\markov}{\small{\mbox{$-\hspace{-1.3mm} \circ \hspace{-1.3mm}-$}}}

\long\def\symbolfootnote[#1]#2{\begingroup\def\thefootnote{\fnsymbol{footnote}}\footnote[#1]{#2}\endgroup}

\allowdisplaybreaks

\begin{document}

\title{Efficient Quantum Polar Codes Requiring\\ No Preshared Entanglement}

\author{Joseph M. Renes~\IEEEmembership{Member,~IEEE}, David Sutter~\IEEEmembership{Student Member,~IEEE},\\ Fr\'ed\'eric Dupuis~\IEEEmembership{Member,~IEEE}, and Renato Renner~\IEEEmembership{Member,~IEEE}}
\maketitle
\begin{abstract}
We construct an \emph{explicit} quantum coding scheme 
which achieves a communication rate not less than the {coherent information} when used to transmit quantum information over a noisy quantum channel. 
For Pauli and erasure channels we also present \emph{efficient}  encoding and decoding algorithms for this communication scheme based on polar codes (essentially linear in the blocklength), but which do not require the sender and receiver to share any entanglement before the protocol begins. 
Due to the existence of degeneracies in the involved error-correcting codes it is indeed possible that the rate of the scheme exceeds the coherent information. We provide a simple criterion which indicates such performance.  
Finally we discuss how the scheme can be used for secret key distillation as well as private channel coding.
\end{abstract}


\symbolfootnote[0]{
J.M.\ Renes, D.\ Sutter and R.\ Renner are with the Institute for Theoretical Physics, ETH Zurich, Z\"urich 8092, Switzerland (e-mail: renes@phys.ethz.ch; suttedav@phys.ethz.ch; renner@phys.ethz.ch).

F.\ Dupuis is with Faculty of Informatics, Masaryk University, Brno, Czech Republic (e-mail: dupuis@fi.muni.cz).

This work was supported by the Swiss National Science Foundation (through the National Centre of Competence in Research `Quantum Science and Technology' and grant No.~200020-135048) and by the European Research Council (grant No.~258932).

This paper was presented in part at the IEEE Information Theory Workshop (2012) and the IEEE International Symposium on Information Theory (2013).
}

\thispagestyle{empty}

\title{Efficient Quantum Polar Codes Requiring No Preshared Entanglement}

%
%
%

\def\W{\mathsf{W}}

\thispagestyle{empty}
\section{Introduction}
\IEEEPARstart{S}{hannon's} channel coding theorem 
determines the capacity of a classical discrete memoryless channel $\W$ by random coding arguments and finds that it is given by
\begin{equation}
C(\W)=\max \limits_{P_X} I\bigl(X:\W(X)\bigr), \label{eq:shannon}
\end{equation}
where the random variable $X$ describes the input to the channel, and $P_X$ is its probability distribution \cite{shannonshort}.
Analogous random coding arguments for the problem of transmitting quantum information over a memoryless quantum channel $\mathcal{N}^{A'\rightarrow B}$ lead to a communication rate given by
\begin{equation}
\label{eq:cohinf}
Q_1(\mathcal{N}):=\max \limits_{\phi} \CI{A}{B}_{\sigma},
\end{equation}
where the optimization is over all pure, bipartite states $\phi^{AA'}$, $\sigma^{AB}:=\mathcal{N}^{A'\rightarrow B}(\phi^{AA'})$ and $I(A\rangle B)_{\rho}:=H(B)_{\rho}-H(AB)_{\rho}$ is the \emph{coherent information} and $H$ the von Neumann entropy~\cite{lloyd_capacity_1997,shor_quantum_2002, devetak_private_2005,schumacher_sending_1997,barnum_information_1998,barnum_quantum_2000}. 
It has been shown that $Q_1(\mathcal{N})$ is not generally optimal~\cite{shorshort} and that the quantum capacity is given by its \emph{regularization} 
\begin{equation}
Q(\mathcal{N})=\lim \limits_{k\to \infty} \frac{1}{k} Q_1\!\left(\mathcal{N}^{\otimes k}\right).
\end{equation}
Notwithstanding the complications surrounding the regularized expression, it is already difficult to construct \emph{explicit} coding schemes that achieve the coherent information of an arbitrary quantum channel, and the task becomes that much harder if we also ask for \emph{efficient} encoding and decoding. 
Here, a task is termed efficient if it can be completed in a number of steps scaling essentially linearly in the input size, not just polynomially as in the complexity-theoretic setting.
Until very recently, almost nothing was known about explicit, efficient, \emph{provably} capacity-achieving classical error-correcting codes, to say nothing of the quantum case.

Polar codes, introduced in 2008 by Ar{\i}kan~\cite{arikan09}, are the first family of classical error-correcting codes which both provably achieve the \emph{symmetric} (i.e., $X\sim$ uniform in \eqref{eq:shannon}) classical capacity for any discrete memoryless channel and have an essentially linear encoding and decoding complexity. These codes have been generalized to the quantum setup. Wilde and Guha adapted polar codes to transmit classical information over quantum channels~\cite{wilde_polar_2011} and gave a scheme for transmitting quantum information over degradable quantum channels~\cite{wildeGuha}, at the cost of an unknown decoding efficiency. Three of us showed in \cite{renesshort} how to achieve the \emph{symmetric} coherent information $(\phi^{AA'}$ a Bell state in \eqref{eq:cohinf}) of any Pauli or erasure channel with efficient encoding and decoding operations. In \cite{renes_polar_2014}, Wilde and Renes extended this method to arbitrary quantum channels and showed that the quantum decoder can be constructed by combining sequential decoders for suitable classical-quantum channels (see also~\cite{wilde_sequential_2013}), but without providing an efficient decoder. 


However, all of these quantum channel coding schemes suffer from two important drawbacks. The first is the need for noiseless entanglement to be shared by the sender and receiver prior to the start of the protocol. Second, the aforementioned protocols only achieve rates given by the symmetric coherent information. More details about entanglement-assisted quantum coding can be found in \cite{todd06}; a precursor in the field of quantum key distribution is~\cite{lo_method_2003}. In this contribution we present an explicit coding scheme that provably achieves the (true) coherent information for an arbitrary quantum channel without using any entanglement assistance. Our protocol is a concatenated two level scheme in which the method of polar coding is employed separately at each level. 
For Pauli or erasure channels we show how to perform efficient encoding and decoding.

Recently, an extension to quantum polar codes that is based on branching MERA codes has been proposed~\cite{poulinMERA,poulinMERA2}. It is claimed (however not yet proven) that these codes do not require entanglement assistance while achieving high rates and an efficient encoding and decoding.

The paper is structured as follows. Section~\ref{sec:def} introduces basic notation and definitions.
As our protocol is based on the use of information reconciliation (IR) and stabilizer-based quantum error-correcting codes, Section~\ref{sec:backir} provides background on these topics. Section~\ref{ap:UR} recalls two general uncertainty relations needed to establish the rate of the coding scheme.
Section~\ref{sec_polpheno}  introduces the classical polarization phenomenon and proves that it also holds for general classical-quantum (cq) states.\footnote{Since, to the best of our knowledge, the polarization phenomenon has thus far only been generalized to symmetric cq-states~\cite{wilde_polar_2011}.}  
Section~\ref{sec:ed} then describes a protocol for the problem of entanglement distillation and Section~\ref{sec:effQPC} shows that it is computationally efficient for Pauli and erasure channels when using quantum polar codes. It is well known that an entanglement distillation featuring only forward classical communication can be turned into a channel coding scheme~\cite{bennett_mixed-state_1996}, and Section~\ref{sec:channelcoding} explains the channel coding view of the scheme in detail. We prove that the communication rate of the channel coding scheme is at least as great as $Q_1(\mathcal{N})$, as defined in \eqref{eq:cohinf}, and present an efficient encoder and decoder based on polar codes for Pauli and erasure channels.  It is possible, but still unproven, that the scheme can achieve rates beyond the coherent information. Section~\ref{sec:beyond} shows, however, that the rate expression reduces to the coherent information when the channel is degradable. Sections~\ref{sec:beyond} and \ref{sec:achievingQ} state two different open problems addressing the question of whether our scheme can achieve rates beyond the coherent information or even achieve the quantum capacity for channels where the coherent information is not optimal. In Section~\ref{sec:skd} it is discussed how the scheme can be used for secret key distillation. 

\section{Background}
\subsection{Notation and Definitions} \label{sec:def}
Let $\mathcal{N}:\mathfrak{T}(\mathcal{H}_A) \to   \mathfrak{T}(\mathcal{H}_B)$ be a quantum channel, where $\mathfrak{T}(\mathcal{H})$ denotes the space of all operators in some Hilbert space $\mathcal{H}$ that are equipped with the trace norm.
We assume that the underlying state spaces $\mathcal{H}_A$ and $\mathcal H_B$ for systems $A$ and $B$ are finite dimensional. Using the Stinespring dilation~\cite{stinespring55} there exists an ``environment'' system $\mathcal{H}_E$ and a partial isometry $U_{BE}: \mathcal{H}_A \to \mathcal{H}_B \otimes \mathcal{H}_E$, such that 
\begin{equation}
\mathcal{N}(\rho)=\Trp{E}{U_{BE} \,\rho \, U_{BE}^{\dagger}}
\end{equation}
for every $\rho \in \mathcal{D}(\mathcal{H}_A$), where $\mathcal{D}(\mathcal{H}_A)$ denotes the space of all density operators on $\mathcal{H}_A$. Then we may define the \emph{complementary channel} to $\mathcal{N}$ as
\begin{equation}
\mathcal{N}_{\setC}(\rho)=\Trp{B}{U_{BE}\, \rho \, U_{BE}^{\dagger}}.
\end{equation}
A quantum channel $\mathcal{N}$ is \emph{more capable} if the quantum capacity of its complementary channel is zero, i.e.,\ $Q(\mathcal{N}_{\setC})=0$. 
Equivalently, the notion of more capable may be formulated as follows. 
For all states $\rho^{A^N}$ on $\mathcal{D}(\mathcal{H}_A^{\otimes N})$ let the classical random variable $X^N$ denote its eigenvalues and $B^N$ and $E^N$ the quantum outputs under the action of $\mathcal{N}^{\otimes N}$. 
Then $\mathcal N$ is more capable when $I(X^N:B^N)\geq I(X^N:E^N)$ for all $N\geq 1$ \cite{watanabe12}.  

If the private capacity of the complementary channel is zero, i.e., $C_{p}(\mathcal{N}_{\setC})=0$, then $\mathcal N$ is said to be \emph{less noisy}.\footnote{Recall that according to \cite{devetak_private_2005,cai04} the private capacity of a quantum channel $\mathcal{N}$ is given by 
\begin{equation}
C_p(\mathcal{N})=\lim \limits_{k\to \infty} \frac{1}{k} \max\limits_{T,X^k} \I{T}{B^k}-\I{T}{E^k},
 \end{equation}
 where $T \markov X^k \markov (B^k,E^k)$ form a Markov chain.
 }
 Equivalently, $\mathcal N$ is less noisy if $\I{T}{B^N}\geq \I{T}{E^N}$, for every $N\geq 1$ and for all distributions on $(T,X^N)$, where $T$ has finite support and $T \markov X^N \markov (B^N,E^N)$ form a Markov chain~\cite{watanabe12}. 

Finally, a quantum channel $\mathcal{N}$ is called \emph{degradable} if there exists a trace-preserving and completely positive ($\TPCP$) map $\mathcal{T}$ such that $\mathcal{T}\circ \mathcal{N}=\mathcal{N}_{\setC}$. 
The concepts of a channel being more capable, less noisy, or degradable were introduced in a classical framework in \cite{korner77} and have recently been generalized to the quantum mechanical setup in \cite{watanabe12}. 
For classical channels, it has been shown that being more capable is a strictly weaker condition than being less noisy, which is again a strictly weaker condition than being degradable \cite{korner77}. For quantum channels it is unknown whether these relations are strict \cite{watanabe12}. However a degradable (quantum) channel $\mathcal{N}$ is less noisy and \emph{a fortiori} more capable. 
In~\cite{watanabe12} it is shown that both the private and quantum capacities of a less noisy channel $\mathcal N$ are equal to the coherent information:  
\begin{equation}
C_p(\mathcal{N})=Q(\mathcal{N})=Q_1(\mathcal{N}),\label{eq:lessnoisyrates}
\end{equation}
a relation that was first shown for degradable channels \cite{shor2short,smi08}.

\subsection{Information Reconciliation \& Stabilizer Codes}\label{sec:backir}
In this subsection we recall some details about the information reconciliation protocols which form the basis of our coding schemes. 
The form of information reconciliation needed here is an instance of data compression in which the decoder has access to quantum side information. 
Consider a classical-quantum state of a random variable $Z$ and a quantum system $B$, which may be written $\psi^{ZB}=\sum_z p_z \ketbra{z}^Z\otimes\varphi_z^B$, where $p_z$ is the probability distribution of $Z$ and $\varphi_z$ are arbitrary states. 
The goal of information reconciliation is to compress $Z$, i.e.\ apply some function to $Z$, so that $Z$ itself can be reconstructed (with high probability) by a decoder with access to $B$ and the compressed output. 

In quantum language, the decoding step may be regarded as a measurement of the systems available to the decoder. 
The probability of error in determining the value of a random variable $Z$ given measurement of some quantum system $B$ can be expressed as 
\begin{equation}
p_{\rm{err}}\left(Z^A|B\right)_\psi:=1-\max \limits_{\mathcal{M}_Z}\sum_{z} p_z \Trs\left[\Lambda_{z}^B \varphi_z^B \right],
\end{equation}
where the maximum is taken over all measurements $\mathcal{M}_Z$ with elements $\Lambda_z^B$. 

Devetak and Winter showed that in the case of $N\rightarrow \infty$ copies of $\psi^{ZB}$, there exist compression functions outputting roughly $NH(Z|B)_\psi + o(N)$ bits which suffice to reliably determine $Z^N$ \cite{devetak03}.
Specifically, calling $B_C$ the classical compressor output, they constructed a sequence of measurements such that $p_{\rm err}(Z^N|B^NB_C)_{\psi^{\otimes N}}\rightarrow 0$ as $N\rightarrow \infty$. 
Moreover, this compression rate was shown to be optimal. 

We shall be interested in information reconciliation of classical information that represents the outcomes of measurements on a quantum system in one of two complementary bases, the ``amplitude'' basis and the ``phase'' basis. 
Any basis $\{\ket{z}\}_{z=0}^{d-1}$ of a $d$-dimensional quantum system may be chosen to be the amplitude basis, and for concreteness we then take the phase basis to be comprised of elements $\ket{\tilde x}=\frac1{\sqrt{d}}\sum_{z=0}^{d-1}\omega^{xz}\ket{z}$, where $\omega=e^{2\pi i/d}$. 
Although the coding scheme we describe below can be made to work for any finite $d$, there is no real loss of generality in setting $d=2$, which we do henceforth.

For the purposes of information reconciliation in this context, it is convenient that the random coding argument can be specialized to random linear functions (as, for instance, in \cite{renes3short}). 
Then, since the compressed output is a linear function of the input, we may regard it as resulting from the measurement of a corresponding set of stabilizer observables. 
Here we recall the basic facts of the stabilizer formalism; for more details see \cite{nielsen}.
For instance, the result of computing the parity of the outcome of measuring 3 qubit systems in the amplitude basis can just as well be thought of as the result of measuring the observable $Z\otimes Z\otimes Z$, where $Z={\rm diag}(1,-1)$. 
In general, any function of the form $f(z^N)=v^N\cdot z^N$ corresponds to the stabilizer operator $Z^{v_1}\otimes Z^{v_2}\otimes\cdots\otimes Z^{v_N}$.
An arbitrary linear function is just a sequence of functions of the form just considered, and therefore any linear function corresponds to a sequence of stabilizer operators. 

Linear functions of phase basis measurments can be similarly regarded as measurements of stabilizer operators using the operator $X=\left(\begin{smallmatrix} 0 & 1\\ 1&0\end{smallmatrix}\right)$. 
Since the operators $X$ and $Z$ anticommute, stabilizers corresponding to vectors $u^N$ and $v^N$ commute when $u^N\cdot v^N=0$. 
Thus, it is possible to consider simultaneously measuring a compressed output for both the amplitude and phase bases. Such a set of commuting stabilizer operators constitutes a CSS code.

The measurement of any set of stabilizer operators can be accomplished by applying a unitary transformation and then performing appropriate measurements on individual qubits.
Each stabilizer measurement necessitates the measurement of one qubit, so that not every qubit is necessarily measured after the transformation.
Due to the linear structure of stabilizer observables, this unitary transformation is simultaneously a mapping of amplitude basis states to themselves, as well as a mapping of phase basis states to themselves.

Note that, although we will consider using CSS codes for information reconciliation in the following sections, the code structure is only used by the encoder.
Stabilizer measurements for recovery operations are not generally used by the decoder.
Instead, the coding techniques for information reconciliation make use of more general measurements, such as the pretty-good measurement.

\subsection{Generalized Uncertainty Relations} \label{ap:UR}
The coding scheme introduced in this paper consists of two layers, both performing an information reconciliation operation in complementary bases.\footnote{Recall that we denote by $X^A$ and $Z^A$ an operator $X$ and $Z$ acting on a system A.} 
In the derivation of the exact rate expression (cf.\ Theorem~\ref{thm:ratescheme}) as well as in the formulation of the two open problems (cf. Sections~\ref{sec:beyond} and \ref{sec:achievingQ}) we will use entropic uncertainty relations that arise when measuring a system in complementary bases. This subsection gives an overview about the uncertainty relations we need to prove several properties of the scheme.

\begin{mylem}[Renes \& Boileau~\cite{renes2short}]
\label{lem:2pcertainty}
Suppose $\psi^{AB}$ is a bipartite state for which $H(Z^A|B)_\psi=0$. Then
\begin{align}
\Hc{X^A}{B}_\psi=\log \dim(A)+\Hc{A}{B}_\psi.
\end{align}
\end{mylem}
\begin{IEEEproof}
Let $R$ be a purification of $\psi^{AB}$; expressing the $A$ system in the amplitude basis, we may write
\begin{align}
\ket{\psi}^{ABR}=\sum_z \sqrt{p_z}\ket{z}^A\ket{\varphi_z}^{BR}
\end{align}
for some probability distribution $p_z$ and normalized states $\ket{\varphi_z}^{BR}$. 
Since $H(Z^A|B)=0$, the states $\varphi_z^B$ are completely distinguishable. Thus, there exists a measurement which precisely determines the value of $z$ given $B$. Let $U^{B\rightarrow BC}$ be a coherent implementation (Stinespring dilation) of this measurement, in which the measurement result is stored in system $C$. Then define
\begin{align}
\ket{\psi'}^{ABCR}:= \,\, & U^{B\rightarrow BC}\ket{\psi}^{ABR}\\
=\,\, & \sum_z \sqrt{p_z}\ket{z}^A\ket{z}^{C}\ket{\varphi_z}^{BR}.
\end{align}
As $U$ is a partial isometry, $H(X^A|B)_\psi=H(X^A|BC)_{\psi'}$ and $H(A|B)_\psi=H(A|BC)_{\psi'}$. Expressing $A$ in $\ket{\psi'}$ in the phase basis, we have
\begin{align}
\ket{\psi'}^{ABCR}&=\frac{1}{\sqrt{d}}\sum_{x,z}\sqrt{p_z}\,\omega^{-xz}\ket{\widetilde{x}}^A \ket{z}^C\ket{\varphi_z}^{BR}\\
&=\frac{1}{\sqrt{d}}\sum_{x,z}\ket{\widetilde{x}}^A (Z^{-x})^C\ket{\psi}^{CBR},
\end{align}
where $d={\rm dim}(A)$. Using the chain rule we can write 
\begin{align}
&H(X^A|BC)_{\psi'}\nonumber \\
&\hspace{12mm}=H(X^ABC)_{\psi'}-H(BC)_{\psi'}\\
&\hspace{12mm}=H(BC|X^A)_{\psi'}+H(X^A)_{\psi'}-H(BC)_{\psi'}.
\end{align}
Clearly $H(X^A)_{\psi'}=\log {\rm dim}(A)$, while $H(BC|X^A)_{\psi'}=H(AB)_\psi$ since all the $BC$ marginals conditioned on $X^A$ are related by the unitary $Z^{-x}$ and the $BC$ marginal for $X^A=0$ is $\ket{\psi}^{CBR}\simeq \ket{\psi}^{ABR}$. Finally, $H(BC)_{\psi'}=H(B)_\psi$ again by the isometry property of $U$. Therefore,
\begin{align}
\Hc{X^A}{BC}_{\psi'}&=\Hh{BC}_\psi+\log {\rm dim}(A)-\Hh{B}_\psi,
\end{align}
which by the chain rule is the expression we set out to prove. 
\end{IEEEproof}

\begin{mylem}[Renes \& Boileau~\cite{renes2short}]
\label{lem:3pcertainty}
Suppose $\psi^{ABR}$ is a pure state for which $H(Z^A|B)_\psi=0$. Then
\begin{align}
\Hc{Z^A}{R}_\psi+\Hc{X^A}{B}_\psi=\log \dim A.
\end{align}
\end{mylem}
\begin{IEEEproof}
The statement follows from Lemma~\ref{lem:2pcertainty} and the fact that $H(A|B)_\psi=H(Z^A|B)_\psi-H(Z^A|R)_\psi$ for any state $\psi$. To see the latter expression, use the chain rule to find
\begin{align}
&\Hc{Z^A}{B}_{\psi}-\Hc{Z^A}{R}_{\psi} \nonumber\\
&\hspace{6mm}= \Hh{Z^AB}_{\psi}-\Hh{B}_{\psi}-\Hh{Z^AR}_{\psi}+\Hh{R}_{\psi}\\
&\hspace{6mm}= \Hc{B}{Z^A}_{\psi}-\Hc{R}{Z^A}_{\psi}+\Hc{A}{B}_{\psi}.
\end{align}
Since $ABR$ is pure, the $BR$ marginals conditioned on a projective measurement of $A$ are, too, and therefore $H(B|Z^A)_{\psi}=H(R|Z^A)_{\psi}$. 
\end{IEEEproof}


\section{Polarization Phenomenon} \label{sec_polpheno}
In this section, we introduce the \emph{polarization phenomenon} discovered by Ar{\i}kan, which has been used to construct codes (called \emph{polar codes}) that can be used for channel \cite{arikan09} as well as for source coding \cite{arikan10}. We show how it can be generalized to the setup of classical-quantum (cq) states.

\subsection{Classical Polarization Phenomenon}  \label{ap:polarCodes}
Polar codes have several desirable attributes~\cite{arikan09,sasoglu09,arikantelatar09,honda12}: they achieve the capacity when used for transmitting information over a discrete memoryless channel (DMC); they can be encoded and decoded efficiently (with a complexity that is essentially linear in the blocklength); the error probability of the efficient decoder decays exponentially in the square root of the blocklength. 

Let $X^N$ be a vector whose entries are i.i.d.\ Bernoulli($p$) distributed for some $p\in [0,1]$ and $N=2^n$ for $n\in \mathbb{Z}^+$. Furthermore, let $U^N=G_N X^N$, where $G_N = \left(\begin{smallmatrix}
1&0 \\ 1 & 1
\end{smallmatrix} \right)^{\otimes \log N}$ denotes the polarization transform and $\W:\mathcal{X}\to \mathcal{Y}$ be a DMC with a binary input alphabet $\mathcal{X}=\{0,1 \}$, an arbitrary output alphabet $\mathcal{Y}$ and transition probabilities $\W(y|x)$ for $x\in \mathcal{X}$ and $y\in \mathcal{Y}$. $\W^N$ denotes the channel corresponding to $N$ uses of $\W$. For $Y^N=\W^N X^N$ and $\epsilon \in (0,1)$, we define the two sets
\begin{align}
\mathcal{R}_{\epsilon}^N(X|Y):=\left\lbrace i \in [N]: H(U_i|U^{i-1} Y^N)\geq 1-\epsilon \right \rbrace
\end{align}
and
\begin{align}
\mathcal{D}_{\epsilon}^N(X|Y):=\left\lbrace i \in [N]: H(U_i|U^{i-1} Y^N)\leq \epsilon \right \rbrace.
\end{align}
The former consists of outputs $U_j$ which are essentially uniformly random, even given all previous outputs $U^{j-1}$ as well as $Y^N$, while the latter set consists of the essentially deterministic outputs. The polarization phenomenon is that essentially all outputs are in one of these two subsets, and their sizes are given by the conditional entropy of the input $X$ given $Y$.  

\begin{mythm}[Ar{\i}kan \cite{arikan09,arikan10}] \label{thm:polarizationPhenomenon}
For any $\epsilon \in (0,1)$
\begin{equation}
\left| \mathcal{R}_\epsilon^N(X|Y) \right|= N \Hc{X}{Y} - o(N) 
\end{equation}
and
\begin{equation}
 \left| \mathcal{D}_\epsilon^N(X|Y) \right|=N \bigl(1-\Hc{X}{Y}\bigr) - o(N).
\end{equation}
\end{mythm}

Based on this result, Ar{\i}kan showed in \cite{arikan09} how to construct polar codes having the desirable properties mentioned above.

Non-binary random variables can be represented by a sequence of correlated binary random variables, which are then encoded separately. Correlated sequences of binary random variables may be polarized using a multilevel construction, as shown in \cite{sasoglu09}.\footnote{An alternative approach is given in \cite{abbe11_2,sahebi11}, where the polarization phenomenon has been generalized for arbitrary finite fields. We will however focus on the multilevel construction in this paper.} Given $M$ i.i.d.\ instances of a sequence $X=(X_{(1)},X_{(2)},\dots X_{(K)})$ and possibly a correlated random variable $Y$, the basic idea is to first polarize $X_{(1)}^M$ relative to $Y^M$, then treat $X_{(1)}^MY^M$ as side information in polarizing $X_{(2)}^M$, and so on. More precisely, defining $U^M_{(j)}=G_MX^M_{(j)}$ for $j=1,\dots,K$, we may define the random and deterministic sets for each $j$ as 
\begin{align}
	&\mathcal R_{\epsilon,(j)}^M(X_{(j)}|X_{(j-1)},\cdots, X_{(1)},Y):= \nonumber \\
	&\hspace{0mm}\{i\in[M]: \Hc{U_{(j),i}}{U_{(j)}^{i-1},X_{(j-1)}^M,\cdots, X_{(1)}^M,Y^M}\geq 1-\epsilon\}
\end{align}
and
\begin{align}	
	&\mathcal D_{\epsilon,(j)}^M(X_{(j)}|X_{(j-1)},\cdots, X_{(1)},Y) :=\nonumber \\
	&\hspace{0mm}\{i\in[M]: \Hc{U_{(j),i}}{U_{(j)}^{i-1},X_{(j-1)}^M,\cdots, X_{(1)}^M,Y^M}\leq \epsilon\}.
\end{align}
In principle we could choose different $\epsilon$ parameters for each $j$, but this will not be necessary here. Now, Theorem~\ref{thm:polarizationPhenomenon} applies to the random and deterministic sets for every $j$. The sets $\mathcal R_{\epsilon}^{M}(X|Y)= \{ \mathcal R_{\epsilon,(j)}^M(X_{(j)}|X_{(j-1)}, \ldots,X_{(1)},Y) \}_{j=1}^K$ and $\mathcal D_{\epsilon}^{M}(X|Y)=\{ \mathcal D_{\epsilon,(j)}^M(X_{(j)}|X_{(j-1)},\ldots,X_{(1)},Y) \}_{j=1}^K$ have sizes given by  
\begin{align}
	&|\mathcal{R}_{\epsilon}^{M}(X|Y)| \nonumber \\
	&\hspace{6mm}= \sum_{j=1}^{K} \left| \mathcal R_{\epsilon,(j)}^M(X_{(j)}|X_{(j-1)}, \ldots,X_{(1)},Y) \right|\\
	&\hspace{6mm}=\sum_{j=1}^{K} M  \Hc{X_{(j)}}{X_{(1)},\dots,X_{(j-1)},Y}-o(M)\\
     &\hspace{6mm}=M \Hc{X}{Y}-o(KM),
\end{align}
and 
\begin{align}
	&|\mathcal{D}_{\epsilon}^{M}(X|Y)| \nonumber \\
	&\hspace{3mm}= \sum_{j=1}^{K} \left| \mathcal D_{\epsilon,(j)}^M(X_{(j)}|X_{(j-1)},\ldots,X_{(1)},Y) \right|\\
	&\hspace{3mm}=\sum_{j=1}^{K} M\left(1-\Hc{X_{(j)}}{X_{(1)},\dots,X_{(j-1)},Y}\right)-o(M)\\
     &\hspace{3mm}=M\left(K-\Hc{X}{Y}\right)-o(KM).
\end{align}
In the following we will make use of both the polarization phenomenon in its original form, Theorem~\ref{thm:polarizationPhenomenon}, and the multilevel extension. To simplify the presentation, we denote by $\widetilde G_{M}^K$ the $K$ parallel applications of $G_M$ to the $K$ random variables $X^M_{(j)}$. 

\subsection{Polarization Phenomenon for General cq-States} \label{sec:PolCQ}
In this subsection, we generalize the polarization phenomenon to the setup of classical-quantum (cq) states of the form
\begin{equation}
\rho^{XB}= \sum_{x\in \{0,1 \}} p_x \ket{x}\bra{x}^X \otimes \rho_x^B, \label{eq:stateCQ}
\end{equation}
where $p_0,p_1 \geq 0$ such that $p_0+p_1=1$. We note that the special case where $p_0=p_1=\frac{1}{2}$ has been discussed in \cite{wilde_polar_2011}.
We denote the fidelity between the states $\rho_0$ and $\rho_1$ by \cite{uhlmann76,jozsa94}
\begin{equation}
\Fi{\rho_0}{\rho_1}:=\norm{\sqrt{\rho_0} \sqrt{\rho_1}}_1,
\end{equation}
where $\norm{A}_1$ denotes the \emph{trace norm} of the operator $A$, defined as $\norm{A}_1:=\Tr{\sqrt{A^{\dagger}A}}$. Furthermore, for a cq state $\rho^{XB}$ as in \eqref{eq:stateCQ}, we define the quantity
\begin{equation}
\Zc{X}{B}_\rho:=2\sqrt{p_0 p_1}\Fi{\rho_0}{\rho_1}=2\sqrt{p_0 p_1}\norm{\sqrt{\rho_0}\sqrt{\rho_1}}_1. \label{eq:Zdef}
\end{equation}
\prlsection{Fidelity Polarization}
Applying Ar{\i}kan's transformation \cite[Section I.B]{arikan09} to the $X$ systems of two independent cq-states $\rho^{XB}$ as defined in \eqref{eq:stateCQ} gives
\begin{align}
	\tilde\rho 
	&=\sum_{u_1u_2} p_{u_1\oplus u_2}p_{u_2}\ketbra{u_1}^{U_1}\!\otimes\! \ketbra{u_2}^{U_2}\!\otimes\! \rho_{u_1\oplus u_2}^{B_1}\!\otimes\!\rho_{u_2}^{B_2}.
\end{align}


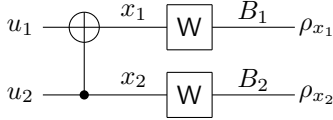
\begin{figure}[!htb]
\centering
\def \xb{0.6} 
\def \yb{0.6} 
\def \y{0.9} 
\def \x{1.1}
\def \s{0.2} 

\def \d{7} 
\def \xs{0.4} 

\def\rad{0.2}
\def \la{0.2} 

\begin{tikzpicture}[scale=1,auto, node distance=1cm,>=latex']
	
     \draw [] (0.5*\x,0) -- (2*\x,0);    
     \draw [] (0.5*\x,-\y) -- (2*\x,-\y); 
     \draw (\x,0) circle [x radius=\rad cm, y radius =\rad cm];
      \fill (\x,-\y) circle (1.75pt);
     \draw [] (\x,\rad) -- (\x,-\y);    
     
      \draw [] (2*\x,0.5*\yb) -- (2*\x+\xb,0.5*\yb);         
      \draw [] (2*\x,-0.5*\yb) -- (2*\x+\xb,-0.5*\yb);          
      \draw [] (2*\x,0.5*\yb) -- (2*\x,-0.5*\yb); 
      \draw [] (2*\x+\xb,-0.5*\yb) -- (2*\x+\xb,0.5*\yb);        
      \node at (2*\x+0.5*\xb,0) {$\W$}; 

      \draw [] (2*\x,0.5*\yb-\y) -- (2*\x+\xb,0.5*\yb-\y);         
      \draw [] (2*\x,-0.5*\yb-\y) -- (2*\x+\xb,-0.5*\yb-\y);          
      \draw [] (2*\x,0.5*\yb-\y) -- (2*\x,-0.5*\yb-\y); 
      \draw [] (2*\x+\xb,-0.5*\yb-\y) -- (2*\x+\xb,0.5*\yb-\y);        
      \node at (2*\x+0.5*\xb,-\y) {$\W$};     
      
     \draw [] (2*\x+\xb,0) -- (3*\x+\xb,0); 
     \draw [] (2*\x+\xb,-\y) -- (3*\x+\xb,-\y);

      \node at (0.5*\x-1.5*\la,0) {$u_1$}; 
      \node at (-1.5*\la+0.5*\x,-\y) {$u_2$};
      \node at (1.6*\x,\la) {$x_1$};              
      \node at (1.6*\x,\la-\y) {$x_2$};        
      \node at (2.5*\x+\xb,\la) {$B_1$};              
      \node at (2.5*\x+\xb,-\y+\la) {$B_2$};   
      \node at (3*\x+\xb+1.5*\la,0) {$\rho_{x_1}$};              
      \node at (3*\x+\xb+1.5*\la,-\y) {$\rho_{x_2}$};

\end{tikzpicture}
\caption{\small Notation used to derived the fidelity polarization phenomenon (cf.\ Proposition~\ref{prop:StateGen}), including the classical-quantum channel $\W: x \to \rho_x$.}
\label{fig:statePol}
\end{figure}

\begin{myprop} \label{prop:StateGen}
\begin{align}
\Zc{U_2}{U_1 B_1 B_2}_{\tilde\rho}&=\Zc{X}{B}_\rho^2 \quad \textnormal{and} \label{eq:1stStat}\\
\Zc{U_1}{B_1 B_2}_{\tilde \rho}&\leq 2 \Zc{X}{B}_\rho-\Zc{X}{B}_\rho^2 . \label{eq:2ndStat}
\end{align}
\end{myprop}
\begin{IEEEproof}
Note that this proof is a generalization of the proof for the symmetric case (i.e.,\ $p_0=p_1=\tfrac{1}{2}$) given by Wilde and Guha \cite[Proof of Proposition 9]{wilde_polar_2011}. We first prove \eqref{eq:1stStat}.
\begin{align}
&\Zc{U_2}{U_1 B_1 B_2}_{\tilde\rho} \nonumber \\
&=2\sqrt{p_0p_1}\,F\left(\sum_{u_1}p_{u_1}\ketbra{u_1}\otimes \rho_{u_1}\otimes\rho_0,\right. \nonumber \\
&\hspace{25mm}\left.\sum_{u_1}p_{u_1\oplus1}\ketbra{u_1}\otimes \rho_{u_1\oplus1}\otimes\rho_1\right)\\
&= 2 \sqrt{p_0 p_1}\,F\left(\sum_{u_1} p_{u_1} \ketbra{u_1} \otimes \rho_{u_1}, \right. \nonumber \\
& \hspace{20mm} \left. \sum_{u_1} p_{u_1\oplus 1} \ketbra{u_1}  \otimes \rho_{u_1\oplus 1}\right)\Fi{\rho_0}{\rho_{1} }\label{eq:mult}\\
&= 4 \sqrt{p_0 p_1}\,\Fi{p_0 \rho_0}{p_1 \rho_1}\Fi{\rho_0}{\rho_{1}} \label{eq:cqFid}\\
&= 4 {p_0 p_1}\, \Fi{\rho_0}{\rho_1}^2\\
&= \Zc{X}{B}_\rho^2,
\end{align}
where \eqref{eq:mult} uses the multiplicativity of the fidelity under tensor product states (i.e.,\ $F(\rho \otimes \sigma,\tau \otimes \nu)=F(\rho,\tau)F(\sigma,\nu)$) and \eqref{eq:cqFid} uses the following relation for cq-states
\begin{align}
&\Fi{\sum_x p(x) \ket{x} \bra{x} \otimes \rho_x}{\sum_x p(x) \ket{x} \bra{x} \otimes \sigma_x} \nonumber \\
&\hspace{45mm}=\sum_x p(x) \Fi{\rho_x}{\sigma_x}.
\end{align}

We next prove \eqref{eq:2ndStat}. Note that the fidelity can be expressed as the minimum Bhattacharya overlap between distributions induced by a POVM on the states \cite{fuchs99}
\begin{equation}
\Fi{\rho_0}{\rho_1}=\min \limits_{\{ \Lambda_m \}} \sum_m \sqrt{\Tr{\Lambda_m \rho_0} \Tr{\Lambda_m \rho_1}}.
\end{equation}
Let $\Lambda_m$ denote a POVM that achieves the minimum for $Z(X|B)_\rho$, we obtain
\begin{align}
\Zc{X}{B}_\rho&= 2 \sqrt{p_0 p_1} \, \Fi{\rho_0}{\rho_1}\\
&= 2 \sqrt{p_0 p_1} \sum_{m}\sqrt{\Tr{\Lambda_m \rho_0} \Tr{\Lambda_m \rho_1}}.
\end{align}
We can use the POVM $\{\Lambda_l \otimes \Lambda_m \}$ to bound $Z(U_1|B_1B_2)_{\tilde \rho}$:
\begin{align}
& \Zc{U_1}{B_1 B_2}_{\tilde\rho} \nonumber \\
&= 2F\left(\sum_{u_2}p_{u_2}^2\rho_{u_2}\otimes\rho_{u_2},\sum_{u_2}p_{u_2\oplus 1}p_{u_2}\rho_{u_2\oplus 1}\otimes\rho_{u_2}\right)\\
&\leq 2\sum_{l,m}\left(\sum_{u}p_u^2\Tr{(\Lambda_l\otimes \Lambda_m)(\rho_u\otimes\rho_u)} \right. \nonumber \\
&\hspace{13mm} \left.\sum_{u'}p_{u'}p_{u'\oplus 1}\Tr{(\Lambda_l\otimes \Lambda_m)(\rho_{u'\oplus 1}\otimes\rho_{u'})}\right)^{\frac{1}{2}}\\
&=2\sum_{l,m}\left(\sum_{u}p_u\Tr{\Lambda_l\rho_u}\,p_u\Tr{\Lambda_m\rho_u} \right. \nonumber \\
&\hspace{18mm} \left.\sum_{u'}p_{u'\oplus 1}\Tr{\Lambda_l\rho_{u'\oplus 1}}\,p_{u'}\Tr{\Lambda_m\rho_{u'}}\right)^{\frac{1}{2}}.
 \label{eq:stop}
\end{align}
Now introduce the notation $\alpha_m:=p_0 \Tr{\Lambda_m \rho_0}$, $\beta_l:=p_0\Tr{\Lambda_l \rho_0}$, $\gamma_l:=p_1\Tr{\Lambda_l \rho_1}$, and $\delta_m:=p_1 \Tr{\Lambda_m \rho_1}$ and notice that $Z(X|B)_\rho=2\sum_l \sqrt{\beta_l\gamma_l}=2\sum_m\sqrt{\alpha_m\delta_m}$. Then we can write \eqref{eq:stop} as
\begin{align}
&\Zc{U_1}{B_1 B_2}_{\tilde\rho} \nonumber\\
&\leq 2 \sum_{l,m} \sqrt{\alpha_m \beta_l + \gamma_l \delta_m} \sqrt{\alpha_m \gamma_l+\beta_l \delta_m}\\
&\leq 2 \left( \sum_{l,m} \left(\sqrt{\alpha_m \beta_l}+\sqrt{\gamma_l \delta_m} \right) \left(\sqrt{\alpha_m \gamma_l}+\sqrt{\beta_l \delta_m}\right) \right. \nonumber \\
&\hspace{20mm}\left. - 2 \sum_{l,m} \sqrt{\alpha_m \beta_l \gamma_l \delta_m} \right) \label{eq:ArikanIneq}\\
&= 2 \left(\sum_{l,m}\! \left( \left(\alpha_m+\delta_m\right) \sqrt{\beta_l \gamma_l} + \left(\beta_l+\gamma_l\right) \sqrt{\alpha_m \delta_m} \right) \right. \nonumber \\
&\hspace{20mm}\left. -2\sum_l\! \sqrt{\beta_l \gamma_l} \sum_{m}\!\sqrt{\alpha_m \delta_m}\! \right)\\
&= 2  \left(\sum_{l} \sqrt{\beta_l \gamma_l} +  \sum_m \sqrt{\alpha_m \delta_m} \right. \nonumber \\
&\hspace{20mm} \left. -2\sum_l \sqrt{\beta_l \gamma_l} \sum_{m}\sqrt{\alpha_m \delta_m} \right)\\
&=2 \Zc{X}{B}_\rho-\Zc{X}{B}_\rho^2,
\end{align}
where inequality \eqref{eq:ArikanIneq} is due to Ar{\i}kan \cite[Appendix D]{arikan09}.
\end{IEEEproof}

\prlsection{Entropy Polarization}
We begin by bounding the entropy $H(X|B)_\rho$ in terms of the probability of error $P_{\rm e}(X|B)_\rho$ when using $B$ to determine $X$. This is formally defined by
\begin{align}
P_e(X|B)_\rho&:=\min \limits_{0\leq \Lambda \leq \mathbbmss{1}} p_0 \Tr{\Lambda \rho_0} + p_1 \Tr{(\mathbbmss{1}-\Lambda) \rho_1}. \label{eq:inter}
\end{align}
The error probability can be expressed in terms of a trace distance, as follows. From the definition it follows immediately that 
\begin{align}
	P_e(X|B)_\rho&=\min \limits_{0\leq \Lambda \leq \mathbbmss{1}} p_1 + \Tr{\Lambda(p_0 \rho_0 - p_1 \rho_1)} \label{eq:min?}\\
	&=\min \limits_{0\leq \Lambda \leq \mathbbmss{1}} p_0 - \Tr{\Lambda(p_0 \rho_0 - p_1 \rho_1)}.\end{align}
Letting $\Gamma:=p_0 \rho_0 - p_1 \rho_1$, the minimum of \eqref{eq:min?} is achieved for $\Lambda=\{\Gamma \}_{-}$, the projector onto the negative part of $\Gamma$. We thus have $P_e(X|B)_\rho=p_1 + \Tr{\{\Gamma\}_{-}\Gamma}=p_0-\Tr{\{\Gamma\}_+\Gamma}$.
Averaging the two different expressions leads to
\begin{align}
P_e &= \frac{1}{2}\left(p_1 + \Tr{\{\Gamma\}_{-}\Gamma} +p_0 - \Tr{\{\Gamma \}_{+} \Gamma} \right)\\
       &= \frac{1}{2} - \frac{1}{2} \norm{\Gamma}_1\\
       &=\frac{1}{2} -\frac{1}{2}\norm{p_0 \rho_0 - p_1 \rho_1}_1.
\end{align}
Now we have
\begin{mylem} \label{lem:bnd1}
$-\log\bigl(1-P_e(X|B)_\rho\bigr) \leq \Hc{X}{B}_{\rho} \leq \Hb\bigl(P_e(X|B)_\rho\bigr)$.
\end{mylem}
\begin{IEEEproof}
Let $\hat X$ denote a guess of $X$ generated by an optimal measurement, the data processing inequality ensures that $H(X|B)_{\rho}\leq H(X|\hat X)_{\rho}$. Using Fano's inquality, i.e.,\ $H(X|\hat X)_{\rho} \leq \Hb(P_e)+P_e \log(|\mathcal{X}|-1) = \Hb(P_e)$ proves the upper bound.
To prove the lower bound we use $H_{\min}(X|B)_{\rho}=-\log(1-P_e)$ \cite[Section IC]{koenig09}. According to \cite[Proposition 4.3]{tomamichel_phd}, $H(X|B)_{\rho} \geq H_{\min}(X|B)_{\rho}$ which proves the assertion.
\end{IEEEproof}
We next bound $Z(X|B)_\rho$ in terms of $P_e$.
\begin{mylem} \label{lem:bnd2}
Let $\rho_0$ and $\rho_1$ be two arbitrary density operators and $p_0,p_1\geq 0$ such that $p_0+p_1=1$. The parameter  $Z(X|B)_{\rho}:=2\sqrt{p_0 p_1}\norm{\sqrt{\rho_0}\sqrt{\rho_1}}_1$  can be bounded in terms of error probability $P_e:=\tfrac{1}{2}-\tfrac{1}{2}\norm{p_0 \rho_0 - p_1 \rho_1}_1$ as
\begin{equation}
2P_e \leq \Zc{X}{B}_{\rho} \leq \sqrt{1-(1-2P_e)^2}.
\end{equation}
\end{mylem}
 \begin{IEEEproof}
 According to \cite[Lemma A.2.4]{renner_phd} we have $\norm{p_0 \rho_0 -p_1 \rho_1}_1^2 \leq 1-Z(X|B)_{\rho}^2$. Lemma~A.2.6 of \cite{renner_phd} ensures that $1-Z(X|B)_{\rho}\leq \norm{p_0 \rho_0 - p_1 \rho_1}$. By definition of $P_e$ we have $\norm{p_0 \rho_0 - p_1 \rho_1}=1-2P_e$ which completes the proof.
 \end{IEEEproof}
 Using Lemma~\ref{lem:bnd1} and \ref{lem:bnd2} we can bound $H(X|B)_{\rho}$ in terms of $Z(X|B)_{\rho}$.
\begin{myprop} \label{prop:PolEntr}
For $\rho^{XB}= \sum_{x\in \{0,1 \}} p_x \ket{x}\bra{x}^X \otimes \rho_x^B$ and $Z(X|B)_{\rho}:=2\sqrt{p_0 p_1}\norm{\sqrt{\rho_0}\sqrt{\rho_1}}_1$, we have
\begin{align}
&1-\log\left(1+\sqrt{1-\Zc{X}{B}_{\rho}^2} \right) \nonumber \\
&\hspace{25mm}\leq \Hc{X}{B}_{\rho} 
\leq \Hb\left(\frac{1}{2}\Zc{X}{B}_{\rho} \right). \label{eq:LBandUB}
\end{align}
\end{myprop}
\begin{IEEEproof}
Follows immediately from Lemma~\ref{lem:bnd1} and \ref{lem:bnd2}.
\end{IEEEproof}
Proposition~\ref{prop:PolEntr} serves the purpose of showing that $H(X|B)_{\rho}$ is near $0$ or $1$ if and only if $Z(X|B)_{\rho}$ is near $0$ or $1$, respectively, i.e.,\ $H(X|B)_{\rho}$ and $Z(X|B)_{\rho}$ polarize simultaneously. This is visualized in Figure~\ref{fig:entropyPol}.

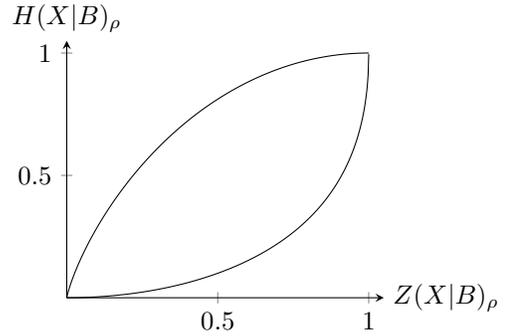
\begin{figure}[!htb]
\centering
\begin{tikzpicture}[]
\begin{axis}[domain=0:1.0,
   xlabel={$Z(X|B)_{\rho}$},
   ylabel={$H(X|B)_{\rho}$},
    xmin=0,
    xmax=1.05, 
    ymin=0,
    ymax=1.05,
   samples=500,
   axis x line=center,
  axis y line=center,
  xtick={0,0.5,1},
  ytick={0,0.5,1},
  xlabel style={right},
  ylabel style={above},
  height=5cm,]
  
\addplot [] {1-ln(1+sqrt(1-x^2))/ln(2)};
\addplot [] {-0.5*x*ln(0.5*x)/ln(2)-(1-0.5*x)*ln(1-0.5*x)/ln(2)};
\end{axis}
\end{tikzpicture}
\caption{\small The lower and upper bound for $\Hc{X}{B}_{\rho}$ given in \eqref{eq:LBandUB}. It shows that $H(X|B)_{\rho}$ and $Z(X|B)_{\rho}$ polarize simultaneously.}
\label{fig:entropyPol}
\end{figure}

\section{Entanglement Distillation}
\label{sec:ed}
Inspired by previous work in a purely classical scenario~\cite{suttershort}, we consider a \emph{concatenated} entanglement distillation scheme based on CSS codes. The explicitly concatenated structure differentiates our approach from that of Devetak and Winter~\cite{devetak_distillation_2005}, based on Devetak's CSS-like approach for channel coding~\cite{devetak_private_2005}.

\subsection{Protocol} \label{ssec:EDscheme}
The scheme consists of an inner layer which performs information reconciliation (IR)
in the amplitude, or $Z$-basis and an outer layer which performs information reconciliation in the phase, or $X$-basis. 
Each layer utilizes a quantum stabilizer code, and together the amplitude and phase codes form a CSS quantum error-correcting code. 
Information reconciliation at the inner layer is performed on $M$ independent blocks, each consisting of $L$ input systems. 
Due to this two-level structure the scheme has a blocklength $N=LM$. 
Letting $K$ denote the number of unmeasured outputs per amplitude block, Figure~\ref{fig:scheme} depicts the case $L=4$, $M=2$, and $K=2$. 
The unmeasured qubits after the amplitude information reconciliation are forwarded to the phase information reconciliation block.
\begin{figure}[tb]
\hspace{-16.5mm}
\scalebox{0.65}{
\input{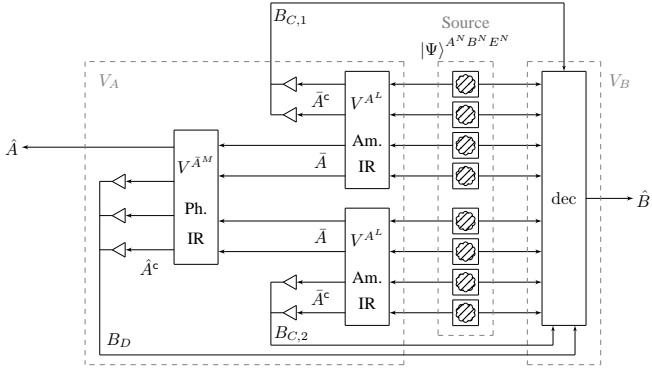}}
\caption{The entanglement distillation scheme for $L=4$, $M=2$ and $K=2$. A source (middle) produces states $\ket{\Psi}$ and distributes the $A$ subsystems to Alice (left), the $B$ subsystems to Bob (right), and retains the $E$ subsystems. Alice performs the amplitude information reconciliation (IR) transformation $M$ times, measures part of the output with respect to the amplitude basis and sends the outcomes to Bob over a classical channel. The non-measured qubits are fed to an IR operation in the complementary phase basis and part of the outcome is again measured and sent to Bob. Using Alice's classical information, Bob runs a decoder such that his and Alice's outcome---described by the systems $\hat A$ and $\hat B$---are a good approximation to maximally entangled qubits.}
\label{fig:scheme}
\end{figure}

To explain the scheme in more detail, we start with a single bipartite system $\psi^{AB}$ shared by Alice and Bob, with $A$ a qubit. Purifying $\psi^{AB}$ and expressing $A$ in the amplitude basis gives 
\begin{equation}
\ket{\psi}^{ABE}=\sum_{z\in \{0,1\}} \sqrt{p_z} \ket{z}^A \ket{\varphi_z}^{BE}, \label{eq:psi}
\end{equation}
where $p_z$ is some probability distribution and $\{\ket{\varphi_z}^{BE}\}$ some set of normalized states, not necessarily orthogonal. The input to each block of the first layer will be the state 
\begin{equation}
\ket{\Psi}^{A^LB^LE^L}=\left(\ket{\psi}^{ABE}\right)^{\otimes L}.
\end{equation}

The amplitude stabilizer code is chosen so that Bob can determine $z^L$, with  probability exceeding $1-\epsilon_1$ using his systems $B^L$ if he is also supplied with the syndromes of the code. These are determined by Alice and transmitted to him over a public classical channel.  Given a particular $z^L$, denote by $\bar z^\setC$ the syndrome and $\bar z$ the encoded information. Since stabilizer codes are linear codes, $z^L$ determines $(\bar z,\bar z^\setC)$ and \emph{vice versa}. Moreover, for every stabilizer code on $L$ systems there exists a unitary operation which maps the stabilizer and encoded operators to physical qubits. Call this unitary $V^{A^L}$; after applying it, Alice need only measure certain subsystems to generate the syndrome. Let $\bar A^{\setC}$ be the systems which are then measured to yield the syndromes and $\bar A$ the remaining systems (corresponding to encoded qubits).  
After applying the unitary, the joint state becomes
\begin{align}
&\ket{\Psi_1}^{A^LB^LE^L} \nonumber \\
&=V^{A^L}\ket{\Psi}^{A^LB^LE^L}\\
&=\sum_{(\bar z, \bar z^{\setC}) \in \{0,1 \}^L} \sqrt{p_{z^L(\bar z, \bar z^{\setC})}} \ket{\bar z}^{\bar A} \ket{\bar z^{\setC}}^{\bar A^{\setC}} \ket{\varphi_{z^L(\bar z, \bar z^{\setC})}}^{B^L E^L}.
\end{align}
 Sending $\bar z^{\setC}$ to Bob can be modeled as copying $\bar z^{\setC}$ to a register $B_C$ he controls, plus another one for the environment, $E_C$  (for notational simplicity, we suppress the dependence of $z^L$ on $(\bar z,\bar z^{\setC})$):
\begin{align}
 &\ket{\Psi_2}^{A^L B^L B_C E^LE_C} \nonumber \\
 &= \sum_{(\bar z, \bar z^{\setC}) \in \{0,1 \}^L} \sqrt{p_{z^L}} \ket{\bar z}^{\bar A} \ket{\bar z^{\setC}}^{\bar A^{\setC}} \ket{\bar z^{\setC}}^{B_C}\ket{\bar z^{\setC}}^{E_C}\ket{\varphi_{z^L}}^{B^L E^L}.
\end{align}
Bob's decoding operation attempts to determine $\bar z$ using the information contained in systems $B^L$ and $B_C$.
It can be thought of as a measurement on $B^L$ conditioned on the value in the register $B_C$; call its elements $\Lambda_{\bar z;\bar z^\setC}^{B^L}$. Performed coherently, it stores the result $z^L$ in an extra ancillary system, say $B_C'$. The post-measurement state is then 
\begin{align}
&\ket{\Psi_3}^{A^LB^LB_CB_C'E^LE_C} \nonumber \\
&=\!\!\!\!\sum_{(\bar z, \bar z^{\setC})\in\{0,1\}^L }\!\!\sum_{\bar z'}\sqrt{p_{z^L}}\ket{\bar z}^{\bar A}  \ket{\bar z^{\setC}}^{\bar A^{\setC}} \ket{\bar z^{\setC}}^{B_C}\ket{\bar z'}^{B_C'} \nonumber \\
&\hspace{35mm}\sqrt{\Lambda_{\bar z';\bar z^\setC}^{B^L}}\ket{\varphi_{z^L}}^{B^L E^L}\ket{\bar z^{\setC}}^{E_C}.
\end{align}
Regarding the pair $(B_CB_C')$ as the system $C^L$, a simple fidelity calculation shows that $\langle \Psi_3 | \hat\Psi_3\rangle\geq 1-\epsilon_1$, where 
\begin{align}
	&\ket{\hat\Psi_3}^{A^LB^LC^LE^LE_C} \nonumber \\
	&=\!\!\!\!\sum_{(\bar z, \bar z^{\setC}) \in \{0,1\}^L}\!\!\!\!  \sqrt{p_{z^L}} \ket {(\bar z, \bar z^{\setC})}^{A^L}\ket{\varphi_{z^L}}^{B^L E^L}\! \ket{(\bar z, \bar z^{\setC})}^{C^L}\ket{\bar z^{\setC}}^{E_C}.
\end{align}

The outer layer performs phase information reconciliation on $M$ instances of the $\bar A$ systems of the state $\ket{\Psi_3}$, where Bob's side information is given by $B^LC^L$ in each instance. In contrast to the inner layer, here the information to be reconciled is not a bit, but a sequence of bits. Therefore, to use the formalism of stabilizer codes, we either need to consider codes over larger dimension 
or multilevel coding schemes. 
Either would work for our purposes, but for concreteness let us opt for the latter. 
Here, Alice and Bob assemble a block of $M$ systems $\bar A$ and sequentially run blocksize-$M$ phase IR protocols on the first qubits in each of the $\bar A$, then the second, and so forth.
At each step they treat already reconciled systems as side information for the current step. 

Ultimately the effect of this procedure can be regarded, as at the inner layer, as applying a unitary $V^{\bar A^M}\!$ and measuring a subset of the output qubits to obtain the syndromes.  These measurement results are sent to Bob, which is modeled as copying them to a register $B_D$ he controls, plus another one for the environment $E_D$. Remaining at the end of this process are a set of unmeasured qubits, the encoded qubits $\hat A$ of the error-correcting code used in phase information reconciliation. 

\subsection{Reliability \& Rate}
\label{subsec:rr}
Now let us examine the scheme more quantitatively. 
Associated with any set of qubits are a set of $X$ and $Z$ operators acting on these qubits; abusing notation, let us refer to the entire collection of these by, for instance, $X^{\bar A}$ and $Z^{\bar A}$ for the set $\bar A$. 
The amplitude IR protocol is chosen to be $\epsilon_1$-good, i.e.,\  $p_{\rm err}(Z^{A^L}|B^L B_C)_{\Psi_2}\leq \epsilon_1$. Since the scheme uses $M$ independent amplitude information reconciliation blocks, we can use the union bound to write
\begin{equation}
p_{\rm err}\left(Z^{A^N}\middle |B^N B_C^M\right)_{\Psi_2^{\otimes M}}\leq M \epsilon_1. \label{eq:amplitudeIR}
\end{equation}
Sidestepping the details of the multilevel coding for the moment, the phase IR protocol is chosen to have 
 \begin{equation}
  p_{\rm err}\left(X^{\bar{A}^M}\middle|B^NC^N B_D\right)\leq \epsilon_2. \label{eq:phaseIR}
 \end{equation}
Clearly $X^{\hat A}$ (cf.~Figure~\ref{fig:scheme}) is a deterministic function of $X^{\bar A^M}$ due to the action of $V^{\bar A^M}$.  However, since this unitary implements a linear function in the basis conjugate to the amplitude observable $Z^{\bar A^M}$, it also implements a linear function in the amplitude basis itself. (This fact was used to show that Ar{\i}kan's polar encoding circuit is directly useful in the quantum setting in~\cite{renesshort}.) Therefore,  $Z^{\hat A}$ is a deterministic function of $Z^{\bar A^M}$ and hence also of $Z^{A^N}$. 
\begin{align}
&p_{\rm err}\left(X^{\hat A}\middle|B^NC^N B_D\right)\leq \epsilon_2\quad\textnormal{and} \label{eq:rel1}\\
&p_{\rm err}\left(Z^{\hat A}\middle|B^N B_C^M\right)\leq M\epsilon_1. \label{eq:rel2}
\end{align}

These conditions ensure that Alice and Bob share a good approximation to $|\hat A|$  maximally entangled qubit pairs. 
Alice's part of the distillation process (summarized in the left hand side of Figure~\ref{fig:scheme}) can be described by a unitary $U_{A}^{A^N \to \hat A B_C^M B_D E_C^M E_D}$. Bob's part is to decode the state $\ket{\Psi}^{A^N B^N E^N}\!$ using $B^N$ and the side information $B_C^M B_D$ he receives from Alice. Inequalities \eqref{eq:rel1} and \eqref{eq:rel2} together with \cite[Theorem 1]{renes10} ensure that there exists a decoding unitary $U_{B}^{B^N B_C^M B_D \to \hat B}$. 
It is constructed directly from the two IR decoders. 

To make the reliability statement precise, define $V_A:=U_A^{A^N \to \hat A B_C^M B_D E_C^M E_D}$, $V_B:=U_B^{B^N B_C^M B_D \to \hat B}$, introduce 
\begin{equation}
\mathcal{E}(\cdot):=\Trs_{E_C^M E_D E^N}[V_B\bigl(V_A(\cdot)V_A^\dagger\bigr)V_B^\dagger] \label{eq:distillationmap}
\end{equation}
and define $\delta({\psi},{\phi})=\tfrac12\|\psi-\phi\|_1$. 
\begin{myprop} \label{prop:quality}
Let $\ket{\phi}_d^{\hat A \hat B}$ be a maximally entangled state of dimension $d$, where $d=\dim \hat A$. Then
\begin{equation}
\delta\left({\phi}_d^{\hat A \hat B}, \mathcal{E}\bigl({\Psi}^{A^N B^N E^N}\bigr) \right)\leq \sqrt{2\epsilon_2}+\sqrt{2M\epsilon_1}.
\end{equation}
\end{myprop}
\begin{IEEEproof}
This proposition follows immediately from \eqref{eq:rel1}, \eqref{eq:rel2} and \cite[Theorem 1]{renes10}. 
\end{IEEEproof}

The rate of the scheme is defined as the number of output qubits divided by the number of input qubits (cf. Equation~\eqref{eq:ratedef}). 
\begin{mythm} \label{thm:ratescheme}
The rate of the scheme is
\begin{equation}
R  = \frac{1}{L}\CI{\bar A}{B^L C^L}_{\Psi_3'} + \frac{o(L)}{L}, \label{eq:CI_layer} 
\end{equation}
 where  
 \begin{align}
&\ket{\Psi_3'} \nonumber \\
&=\!\!\!\!\!\!\sum_{(\bar z, \bar z^{\setC}) \in \{0,1\}^L}\!\!\!\!\!\!\! \!\!\! \sqrt{p_{z^L(\bar z,\bar z^{\setC})}} \ket {\bar z}^{\bar A}\! \ket{\bar z^{\setC}}^{\bar A^{\setC}}\! \ket{\varphi_{z^L(\bar z,\bar z^{\setC})}}^{B^L E^L}\! \ket{z^L(\bar z,\bar z^{\setC})}^{C^L}\!\!\!. \label{eq:psi3Prime}
\end{align} 
\end{mythm}
\begin{IEEEproof}
The rate of the scheme is
\begin{align}
R:=& \, \frac{\big| \hat A \big|}{N} \label{eq:ratedef} \\
    =& \,\frac{N-M \big| \bar A^{\setC}\big| - \big| \hat A^{\setC}\big|}{N}\\
    =&\, \frac{N-ML H(Z^{A}|B)_{\psi}-MH(X^{\bar A}|B^L C^L)_{\hat\Psi_3}}{N}. \label{eq:notfini}
\end{align}
Since we only care about the $\bar A$ register, the copy of $\bar z^{\setC}$ in $E_C$ does not contribute to the entropy in the final term. Thus,
\begin{equation}
\Hc{X^{\bar A}}{B^L C^L}_{\hat\Psi_3} = \Hc{X^{\bar A}}{B^L C^L}_{\Psi_3'},
\end{equation}
where $\ket{\Psi_3'}$ is as defined in \eqref{eq:psi3Prime}. Using Lemma~\ref{lem:2pcertainty} we obtain
\begin{align}
&\Hc{X^{\bar A}}{B^L C^L}_{\Psi_3'} \nonumber \\
 &\hspace{2mm}= \log|\bar A| + \Hc{\bar A}{B^L C^L}_{\Psi_3'} \\
&\hspace{2mm}= L\left(1-\Hc{Z^A}{B}_{\psi}\right)+\Hc{\bar A}{B^L C^L}_{\Psi_3'} +o(L),
\end{align}
where the final step is ensured by the fact that the code is capacity-achieving. We thus can write \eqref{eq:notfini} as
\begin{align}
R &= -\frac{\Hc{\bar A}{B^L C^L}_{\Psi_3'}}{L} + \frac{o(L)}{L} \\
   & = \frac{\CI{\bar A}{B^L C^L}_{\Psi_3'}}{L} + \frac{o(L)}{L},
\end{align}
which proves the assertion.
\end{IEEEproof}

\begin{mycor} \label{cor:CI}
For $\psi$ as given in \eqref{eq:psi}, the rate is larger than or equal to the coherent information, i.e.,\
\begin{align}
     R&\geq  \max \left \lbrace 0, \CI{A}{B}_{\psi} \right \rbrace.  \label{eq:finalStep}
     \end{align}
\end{mycor}

\begin{IEEEproof}
Recall that the states $\ket{\psi}$ and $\ket{\Psi_3'}$ are defined in \eqref{eq:psi} and \eqref{eq:psi3Prime}. Let 
\begin{equation}
\ket{\psi'}^{ABCE}=\sum_{z \in \{0,1\}} \sqrt{p_z} \ket{z}^A \ket{z}^C \ket{\varphi_z}^{BE}. \label{eq:psiPrimeState}
\end{equation}
By definition (cf.\ \eqref{eq:ratedef}) the rate is non-negative. As explained in the proof of Theorem~\ref{thm:ratescheme} we can write
\begin{align}
R=\frac{N-ML \Hc{Z^{A}}{B}_{\psi}-M\Hc{X^{\bar A}}{B^L C^L}_{\Psi_3'}}{N}.
\end{align}
Using the chain rule we obtain
\begin{align}
&\Hc{X^{\bar A}}{B^L C^L}_{\Psi_3'} \nonumber \\
 &\hspace{4mm}=\Hc{X^{A^L}}{B^L C^L}_{\Psi_3'}- \Hc{X^{\bar A^{\setC}}}{B^L C^L  X^{\bar A}}_{\Psi_3'}\\
&\hspace{4mm}= L \Hc{X^A}{BC}_{\psi'}-\Hc{X^{\bar A^{\setC}}}{B^L C^L X^{\bar A}}_{\Psi_3'}.
\end{align}
We thus have
\begin{align}
     R&=1-\Hc{Z^A}{B}_{\psi}-\Hc{X^A}{BC}_{\psi'} \nonumber \\
     &\hspace{34mm}+\frac{\Hc{X^{\bar A^{\setC}}}{B^L C^L X^{\bar A}}_{\Psi_3'}}{L} \label{eq:bfquestion} \\
        &\geq \max \left \lbrace 0, 1- \Hc{Z^A}{B}_{\psi}-\Hc{X^A}{BC}_{\psi'} \right \rbrace \label{eq:QuestionTight}\\
        &= \max \left \lbrace 0, -\Hc{A}{B}_{\psi} \right \rbrace   \label{eq:finalStep2} \\
        &= \max \left \lbrace 0, \CI{A}{B}_{\psi} \right \rbrace.
     \end{align}
Equality \eqref{eq:finalStep2} holds since
\begin{align}
&1-\Hc{Z^A}{B}_{\psi} - \Hc{X^A}{BC}_{\psi'} \nonumber \\
&\hspace{15mm}= 1-\Hc{A}{B}_{\psi'}-\Hc{X^A}{BC}_{\psi'} \\
&\hspace{15mm}= - \Hc{A}{B}_{\psi'}-\Hc{A}{BC}_{\psi'} \label{eq:URstep} \\
&\hspace{15mm}= -\Hc{A}{B}_{\psi'} - \Hc{C}{AB}_{\psi'}\\
&\hspace{15mm}= -\Hc{AC}{B}_{\psi'}\\
&\hspace{15mm}= -\Hc{A}{B}_{\psi}, 
\end{align}
where \eqref{eq:URstep} uses Lemma~\ref{lem:2pcertainty} and that $H(Z^A|BC)_{\psi'}=0$.
\end{IEEEproof}

Note that the scheme presented above works for any CSS code meeting the two reliability conditions~\eqref{eq:rel1} and \eqref{eq:rel2}. A particularly favorable family of codes are the quantum polar codes. Using these codes for Pauli or erasure channels, we know how to build an efficient encoder and decoder having essentially linear complexity and being reliable for large enough blocklength. This will be explained next.

\section{Using Quantum Polar Codes for Pauli or Erasure Channels}  \label{sec:effQPC}
By using polar codes, Alice and Bob can perform the operations explained in Section~\ref{sec:ed} in a computationally efficient manner for states $\ket{\psi}^{ABE}$ that arise from sending half of an entangled pair through a Pauli or erasure channel.  

\subsection{Code Construction} 
Before the protocol starts one must construct the code, i.e.,\ determine the qubits comprising the systems $\bar A^{\setC}$ at the inner layer and $\hat A^{\setC}$ at the outer layer. (Recall that these are the qubits that are measured by Alice and whose measurement outcomes are sent to Bob.) Constructing the system $\bar A^{\setC}$ can be approximately done in linear time using Tal and Vardy's algorithm \cite{talandvardy10} and its adaptation to an asymmetric setup as explained in \cite{honda12}, or alternatively using the more recent algorithm by Tal \emph{et al.}~\cite{tal12}.

To determine the system $\hat A^{\setC}$ requires more effort. Applying the above algorithm for a ``super-source'' seen by the outer layer will not be efficient in the overall blocklength $N$ since its alphabet size is exponential in $L$. Nonetheless, due to the structure of the inner layer, it is perhaps possible that the method of approximation by limiting the alphabet size~\cite{talandvardy10,tal12} can be extended to this case.

\subsection{Encoding}
As described in Section~\ref{ssec:EDscheme} and Figure~\ref{fig:scheme}, starting with a state $\ket{\Psi}^{A^N B^N E^N}\!$, Alice first applies $M$ times a unitary $V^{A^L}$ to perform amplitude information reconciliation, which is in the specific case of using quantum polar codes $V^{A^L}=\sum_{z^L\in \{0,1 \}^{L}}\ket{G_{L} z^L} \bra{z^L}$, where $G_{L}=G^{\otimes \log{L}}$ and $G=\left(\begin{smallmatrix} 1 & 1\\ 0 & 1
\end{smallmatrix} \right)$ denotes the Ar{\i}kan polar transform \cite{arikan09}. 
Alice measures the frozen qubits with respect to the amplitude basis and sends the outcome to Bob. 
 
Alice next applies another polar transform $V^{\bar A^M}\!$---this time with respect to the phase basis---to the $M$ systems $\bar A$. Here we use a multilevel coding scheme, as mentioned above and described in more detail in Section~\ref{ap:polarCodes}. The unitary $V^{\bar A^M}\!$ applied to elements in the phase basis acts as $(\widetilde G_{M}^K)^T$, i.e.,
\begin{align}
V^{\bar A^M} &= \frac{1}{2^{KM}} \sum_{x,x',z \in \{0,1 \}^{KM}} (-1)^{x' \cdot  \widetilde{G}_{M}^K z + x \cdot z} \ket{\tilde x'} \bra{\tilde x}\\
&= \sum_{x \in \{0,1 \}^{KM}} \ket{(\widetilde G_{M}^K)^T \tilde x}\bra{\tilde x},
\end{align}
where we have used $(\widetilde G_{M}^K )^{-1}=\widetilde G_{M}^K$.
The frozen qubits are measured with respect to the phase basis and its outcomes are sent again to Bob. The remaining qubits form the system $\hat A$.

\begin{mythm} \label{thm:alice}
The encoding of the distillation scheme can be done with $O(N \log N)$ steps.
\end{mythm}
\begin{IEEEproof}
The $M$ polar transformations with respect to the amplitude basis can be performed in $O(M L \log L)$ complexity. The second polar transform, this time with respect to the phase basis, has $O(MK \log(MK))$ complexity, where $K=|\bar A|$. Hence it follows that all the amplitude measurements require $O(M(L-K))$ operations and all phase measurements can be done with $O(MK)$ complexity. From the polarization phenomenon (cf.\ Theorem~\ref{thm:polarizationPhenomenon}) we obtain that $K=O(L)$ which proves the assertion.
\end{IEEEproof}

\subsection{Decoding}
Decoding of the $B^N$ system, with the additional information stored in the $B_C^M B_D$ registers, can be done by combining ideas from \cite{renes10} and \cite{renesshort}.
As shown schematically in Figure~\ref{fig:Bob}, Bob's operation is constructed by using the classical polar decoders for amplitude and phase IR in sequence. Note however that these two decoding tasks are not independent as this would neglect possible correlations between amplitude and phase.  

In the first step Bob performs the amplitude IR decoding operation $\mathcal D_A$ ($M$ times), which corresponds to the standard classical polar decoder as introduced in \cite{arikan09}, and stores the result in an auxiliary system $F_i$, $i\in\{1,\ldots,M\}$. Each instance of $\mathcal D_A$ requires the corresponding frozen information, the values $\bar{z}^\setC$, which is provided in $B_C^M$. 

Bob next performs the phase IR decoding operation $\mathcal D_P$, using the information gained from decoding the first layer. Following \cite{suttershort}, we next show how the standard polar decoding procedure can be adapted for the outer layer of a two-stage polar scheme. 

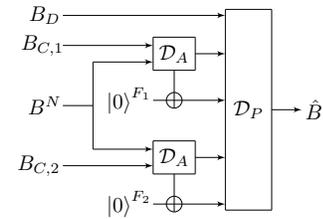
\begin{figure}[!htb]
\hspace{8mm}
\scalebox{0.8}{
\def \x{0.5}
\def \y{1}
\def \ys{0.275} 
\def \yc{0.5} 
\def \xb{0.7} 
\def \xp{0.77} 
\def \xs{0.3} 

\def \la{0.2}
\def \cnotradius{0.13}
\def \radiussmall{0.03}
\def \boxx{0.45}
\def \xspec{0.75} 
\def \shi{0.15} 

\def \tin{0.10} 

\begin{tikzpicture}[scale=1,auto, node distance=1cm,>=latex']
	
     \draw [->] (0,\yc-\y) -- (4*\x+\xb,\yc-\y);     
     \draw [->] (0,-\y) -- (3*\x,-\y);
     \draw [] (0,-2*\y) -- (\x,-2*\y);
     \draw [->] (0,-3*\y) -- (3*\x,-3*\y);

       \draw [] (\x,-2*\y) -- (\x,-\y-\ys);
       \draw [->] (\x,-\y-\ys) -- (3*\x,-\y-\ys);
       \draw [] (\x,-2*\y) -- (\x,-3*\y+\ys);    
       \draw [->] (\x,-3*\y+\ys) -- (3*\x,-3*\y+\ys);

     \draw [draw] (3*\x,-\y+0.5*\ys) -- (3*\x,-\y-1.5*\ys);      
     \draw [draw] (3*\x+\xb,-\y+0.5*\ys) -- (3*\x+\xb,-\y-1.5*\ys);    
      \draw [draw] (3*\x,-\y+0.5*\ys) -- (3*\x+\xb,-\y+0.5*\ys);      
      \draw [draw] (3*\x,-\y-1.5*\ys) -- (3*\x+\xb,-\y-1.5*\ys);  
    \node at (3*\x+0.5*\xb,-\y-0.5*\ys) {$\mathcal{D}_A$};      
 
      \draw [draw] (3*\x,-3*\y+1.5*\ys) -- (3*\x,-3*\y-0.5*\ys);      
     \draw [draw] (3*\x+\xb,-3*\y+1.5*\ys) -- (3*\x+\xb,-3*\y-0.5*\ys);    
      \draw [draw] (3*\x,-3*\y+1.5*\ys) -- (3*\x+\xb,-3*\y+1.5*\ys);      
      \draw [draw] (3*\x,-3*\y-0.5*\ys) -- (3*\x+\xb,-3*\y-0.5*\ys);  
     \node at (3*\x+0.5*\xb,-3*\y+0.5*\ys) {$\mathcal{D}_A$};   
     
     \draw [->] (3*\x,-\y-1.5*\ys-\yc) -- (4*\x+\xb,-\y-1.5*\ys-\yc);
     \draw (3*\x+0.5*\xb,-\y-1.5*\ys-\yc) circle (\cnotradius);
     \draw [draw] (3*\x+0.5*\xb,-\y-1.5*\ys) -- (3*\x+0.5*\xb,-\y-1.5*\ys-\yc-\cnotradius);  
                  
      \node at (3*\x-2*\la,-\y-1.5*\ys-\yc) {$\ket{0}^{F_1}$};       
     \draw [->] (3*\x,-3*\y-0.5*\ys-\yc) -- (4*\x+\xb,-3*\y-0.5*\ys-\yc);
     \draw (3*\x+0.5*\xb,-3*\y-0.5*\ys-\yc) circle (\cnotradius);
     \draw [draw] (3*\x+0.5*\xb,-3*\y-0.5*\ys) -- (3*\x+0.5*\xb,-3*\y-0.5*\ys-\yc-\cnotradius);   
                      
      \node at (3*\x-2*\la,-3*\y-0.5*\ys-\yc) {$\ket{0}^{F_2}$};     
      
   \draw [draw] (4*\x+\xb,\yc-\y+\tin) -- (4*\x+\xb+\xp,\yc-\y+\tin);      
   \draw [draw] (4*\x+\xb,-3*\y-0.5*\ys-\yc-\tin) -- (4*\x+\xb+\xp,-3*\y-0.5*\ys-\yc-\tin);          
   \draw [draw] (4*\x+\xb,\yc-\y+\tin) -- (4*\x+\xb,-3*\y-0.5*\ys-\yc-\tin);         
   \draw [draw] (4*\x+\xb+\xp,-3*\y-0.5*\ys-\yc-\tin) -- (4*\x+\xb+\xp,\yc-\y+\tin);

      \draw [->] (4*\x+\xb+\xp,-2*\y-0.25*\ys) -- (5*\x+\xb+\xp,-2*\y-0.25*\ys);        

       \node at (5*\x+\xb+\xp+\la,-2*\y-0.25*\ys) {$\hat B$};   
      \node at (-1.5*\la,\yc-\y) {$B_D$}; 
      \node at (-1.7*\la,-\y) {$B_{C,1}$};                             
      \node at (-1.5*\la,-2*\y) {$B^N$};      
      \node at (-2*\la,-3*\y) {$B_{C,2}$};  
      \node at (4*\x+\xb+0.5*\xp,-2*\y-0.25*\ys) {$\mathcal{D}_{P}$};

        \draw [->] (3*\x+\xb,-\y-0.5*\ys) -- (4*\x+\xb,-\y-0.5*\ys);   
        \draw [->] (3*\x+\xb,-3*\y+0.5*\ys) -- (4*\x+\xb,-3*\y+0.5*\ys);

\end{tikzpicture}}
\caption{Bob's task in the entanglement distillation process for $M=2$. With the help of the ancilla systems $F_1$ and $F_2$, the classical decoders $\mathcal{D}_A$ and $\mathcal{D}_{P}$ are utilized to distill entanglement.}
\label{fig:Bob}
\end{figure}

\prlsection{Efficient Concatenated Classical Source Coding} 
We start by presenting a concatenated classical coding scheme based on polar codes that can be used for efficient source coding. The scheme has been introduced in \cite{suttershort} for efficient channel coding at the optimal rate.  Let  $\mathcal{I}\subseteq [L]$ denote the indices of the frozen bits at the inner layer and let $K=L-\left| \mathcal{I} \right|$. The set $\mathcal{O}\subseteq [KM]$ denotes the indices of the frozen bits at the outer layer.\footnote{The bits at the outer layer are numbered with respect to the multilevel structure, i.e.,\ each binary decompressor is numbered sequentially.} Figure~\ref{fig:clSourceCoding}, depicts the scheme schematically for the setup of $L=4$, $M=2$, $\mathcal{I}=\{2,3\}$, $K=2$ and $\mathcal{O}=\{2\}$.
In the following we prove that there exists an encoder and decoder that require $O(N \log N)$ steps having an error probability not greater than $O(L 2^{-M^{\beta}})$ for any $\beta<\tfrac{1}{2}$.
\begin{figure}[!htb]
	\hspace{-20mm}
\scalebox{0.75}{
	\input{channelviewCL.tex}}
        \caption{\small The source coding scheme for $L=4$, $M=2$, $\mathcal{I}=\{2,3\}$, $K=2$ and $\mathcal{O}=\{2\}$. A source produces $N$ i.i.d.\ copies of correlated random variables $(X,Y)^N$. The encoder first applies $M$ polarization transforms $G_L$ to $X^L$, obtaining a vector $V^N$. The frozen bits that are determined by $\mathcal{I}$ and denoted by $S_{L-K}^{(M)}$ are sent to the decoder whereas the remaining bits are encoded by another polarization transform---performed in a multilevel construction as explained in Section~\ref{ap:polarCodes}---resulting in a vector $T^{(K)}_M$. The frozen bits of $T^{(K)}_M$ (determined by $\mathcal{O}$) are sent to the decoder again. Using the bits received from the encoder the decoder outputs $\hat T^{(K)}_M$; a guess for $T^{(K)}_M$. }
        \label{fig:clSourceCoding}
\end{figure} 

As depicted in Figure~\ref{fig:clSourceCoding} the encoder consists of two parts. It first applies $M$ identical inner encoding transforms---it performs $M$ times the polar transform $G_L$. The outcome $V^N$ can be classified into two systems determined by the code construction. The frozen bits, denoted by $S_{L-K}^{(M)}$, are sent to the decoder whereas the others are encoded a second time by the outer encoder. The outer encoder applies another polarization transform in the multilevel technique explained in Section~\ref{ap:polarCodes} which outputs $T^{(K)}_M$.\footnote{$T^{(i)}_j$ denotes the $j$th output of the $i$th binary outer encoder in the sequential multilevel structure.} The frozen bits of $T^{(K)}_M$, determined by the code construction, are sent to the decoder again.
\begin{figure}[!htb]
\hspace{-4mm}
\scalebox{0.8}{
\def\gapx{1.7}
\def\gapy{.85}
\def\gaph{0.3}
\def\gaps{0.3}
\def\num{0.25}
\def\rad{0.2}
\def\radd{.23}
\def\spl{.2}

\begin{tikzpicture}[gr/.style={gray!75!black},x=\gapx cm,y=-\gapy cm]

\foreach \x in {1,2} {
  \foreach \y in {1,...,8} {
     }
}
\foreach \x in {1,2} {
  \foreach \y in {1,...,2} {
      }
}

\foreach \x in {1,...,8} {
   \draw (2,\x) -- (4.1,\x);
}

\foreach \x in {0,1} {
   \foreach \y in {0,4} {
   \fill (3.5-\spl+2*\x*\spl,\x+3+\y) circle (1.75pt);
   \draw (3.5-\spl+2*\x*\spl,\x+3+\y) -- (3.5-\spl+2*\x*\spl,\x+1-\radd+\y);
   \draw  (3.5-\spl+2*\x*\spl,\x+1+\y) circle [x radius=\rad cm, y radius =\rad cm];
   }   
   }
   
\foreach \x in {1,...,4} {
   \fill (2.5,2*\x) circle (1.75pt);
   \draw (2.5,2*\x) -- (2.5, 2*\x-1-\radd);
   \draw (2.5, 2*\x-1) circle [x radius=\rad cm, y radius =\rad cm];
}
   
\fill (1.25,8) circle (1.75pt);
\draw (1.25,8) -- (1.25,4-\radd);
\draw (1.25,4) circle [x radius=\rad cm, y radius =\rad cm];
\fill[gray] (1.25-.375,4) circle (1.25pt);
\fill[gray] (1.25-.375,8) circle (1.25pt);

\fill (0.1,5) circle (1.75pt);
\draw (0.1,5) -- (0.1,1-\radd);
   \draw (0.1,1) circle [x radius=\rad cm, y radius =\rad cm];
\fill[gray] (-.15,1) circle (1.25pt);
\fill[gray] (-.15,5) circle (1.25pt);


\draw (-.3,1) -- (2,1);
\draw (-.3,5) -- (2,5);
\draw (1-.3,4) -- (2,4);
\draw (1-.3,8) -- (2,8);

\foreach \x in {1,...,8} {
\node[anchor=west] at (4.1,\x) {$x_\x$};
}

\node[anchor=east] at (2,2) { $s^{(1)}_1$};
\draw[<-] (2,2) -- (2.05,2);
\node[anchor=east] at (2,3) { $s^{(1)}_2$};
\draw[<-] (2,3) -- (2.05,3);
\draw[<-] (2,4) -- (2.05,4);
\draw[<-] (2,1) -- (2.05,1);
\node[anchor=east] at (2,6) { $s^{(2)}_1$};
\draw[<-] (2,6) -- (2.05,6);
\node[anchor=east] at (2,7) { $s^{(2)}_2$};
\draw[<-] (2,7) -- (2.05,7);
\draw[<-] (2,5) -- (2.05,5);
\draw[<-] (2,8) -- (2.05,8);
\fill[gray] (1.75,1) circle (1.25pt);
\fill[gray] (1.75,4) circle (1.25pt);
\fill[gray] (1.75,5) circle (1.25pt);
\fill[gray] (1.75,8) circle (1.25pt);

\node[anchor=east] at (-.3,1) { ${t}^{(1)}_1$};
\draw[<-] (-.3,1) -- (-0.25,1);
\node[anchor=east] at (-.3,5) { ${t}^{(1)}_2$};
\draw[<-] (-.3,5) -- (-0.25,5);
\node[anchor=east] at (1-.3,4) { ${t}^{(2)}_1$};
\draw[<-] (1-.3,4) -- (1-0.25,4);
\node[anchor=east] at (1-.3,8) { ${t}^{(2)}_2$};
\draw[<-] (1-.3,8) -- (1-0.25,8);

\foreach \x in {1,...,8} {
\fill[gray] (2.9,\x) circle (1.25pt);
\fill[gray] (4,\x) circle (1.25pt);
\draw[->] (3.95,\x) -- (3.9,\x);
}

\end{tikzpicture}}
\caption{\small Encoding circuit for the setup $L=4$, $M=2$, $\mathcal{I}=\{2,3\}$, $K=2$ and $\mathcal{O}=\{2\}$. Here $s_j^{(i)}$ denote the frozen bits that are sent to the decoder at the inner layer. We have a single frozen bit at the outer layer, $t_2^{(1)}$, that is sent to the decoder as well. The small gray dots represent variables in the network and correspond to nodes in Fig.~\ref{fig:clDecoder}.}
\label{fig:clEncoder}
\end{figure}
\begin{mylem} \label{lem:ApEncode}
 The classical encoder explained above has complexity $O\left(N \log N \right)$.
\end{mylem}
\begin{IEEEproof}
The inner encoder performs $M$ times a standard polar transform $G_L$ which has been shown to require at most $O(L \log L)$ steps each \cite{arikan09}. 
 The outer encoder consists of $K$ multiplications with the matrix $G_M$, each requiring $O(M\log M)$ operations~\cite{arikan09}.  As justified in Section~\ref{sec:effQPC} and ensured by the polarization phenomenon (cf. Theorem~\ref{thm:polarizationPhenomenon}), $K=O(L)$. We thus conclude that the total encoding requires $O(M L \log L) + O(L M \log M) = O(N \log N)$ steps.
\end{IEEEproof}

The decoding is more challenging.
An important feature of the decoder is that the inner layer decompressors must be interleaved with the outer layer decompressors in order to ensure that all required variables are known at the appropriate steps. To illustrate, we explain in detail how the decoding is done for the setup $L=4$, $M=2$, $\mathcal{I}=\{2,3\}$, $K=2$ and $\mathcal{O}=\{2\}$. The logical structure of the successive cancellation decoder is shown in Figure~\ref{fig:clDecoder}. Figure 10 of~\cite{arikan09} depicts a similar representation of the original successive cancellation decoder. To see the close affinity between the encoding and decoding process, Figure~\ref{fig:clEncoder} visualizes the encoder for the setup defined above. 

Each node in Figure~\ref{fig:clDecoder} is responsible for computing a likelihood ratio (LR) arising during the algorithm; the parameters below each node represent the variables involved in the associated LR computation. Starting from the left we traverse the diagram to the right at whose border we can compute the LRs. Then we transmit the results back to the left. Here $\hat{t}_j^{(i)}$ denotes the $j$th output of the $i$th decompressor at the outer layer and $s_j^{(i)}$ denotes the $j$th frozen input for the $i$th inner encoding block which has been sent to the decoder as explained above.

The decoding begins by activating node $1$, which would like to compute the LR for $T_1^{(1)}$ given $Y_1^8$. For this it needs the LRs for the first inputs to the two super-channels, and so node $1$ activates node $2$, which is responsible for computing the LR for the first input to the first super-channel. This computation proceeds exactly as the usual successive cancellation decoder, recursively combining the LRs of the physical channels by calling node $3$ and then $6$. Assembling their results, node $2$ can compute its LR and transmits the result to nodes $1$ and $16$. Meanwhile, node $1$ has also requested the LR of node $9$, which performs the same calculation as node $2$ for the second super-channel, again forwards the result to nodes $1$ and $16$.  
Now node $1$ is able compute the final desired LR and can therefore guess $\hat{t}_1^{(1)}$. Node $16$ next guesses $\hat{t}_2^{(1)}$, which is easy since this is a frozen bit and therefore available at the decoder (i.e.,\ the decoder sets $\hat t_2^{(1)}=t_2^{(1)}$), completing first decompressor of the outer layer. 

Node $16$ passes control to node $17$ in order to compute the LR for $T^{(2)}_1$.  This requires the LR for second inputs to the two super-channels, so nodes $18$ (and later $21$) are called. Node $18$ finishes the decompression of the first super-channel in the usual way, while node $21$ completes the decompression of the second super-channel. \emph{Neither of these can occur until the first outer layer decompressor is finished}. After the inner layer decompression is complete, node $17$ can guess $\hat{t}_1^{(2)}$ and node $24$ can finally guess $\hat{t}_2^{(2)}$, completing the second decompressor of the outer layer. 
In general, decompression of the $M$ different $k$th inputs at the inner layer has to wait for the $(k-1)$th decompressor to finish at the outer layer. 

\begin{figure}[!htb]
\hspace{-2mm}
\scalebox{0.8}{
\def\gapx{1.7}
\def\gapy{1.2}
\def\gaph{0.3}
\def\gaps{0.3}
\def\num{0.25}
\def\la{0.05}

\begin{tikzpicture}[gr/.style={gray!75!black},x=\gapx cm,y=-\gapy cm]

\foreach \x in {1,2} {
  \foreach \y in {1,...,8} {
     \fill (\x+2,\y) circle (1.75pt);
     }
}
\foreach \x in {1,2} {
  \foreach \y in {1,...,2} {
     \fill (2*\x-2,4*\y-3) circle (1.75pt);
     \fill (\x,4*\y) circle (1.75pt);
      }
}

\foreach \x in {1,...,2} {
   \draw[] (3,\x) -- (4,\x+2) -- (3,\x+2) -- (4,\x) -- (3,\x);
   \draw[] (3,\x+4) -- (4,\x+2+4) -- (3,\x+2+4) -- (4,\x+4) -- (3,\x+4);

   }


\draw[] (0,1) -- (3,1);
\draw[] (2,1) -- (3,2);
\draw[] (0,1) -- (2,5);
\draw[] (2,5) -- (3,6);
\draw[] (0,5) -- (3,5);
\draw[] (2,1) -- (0,5);
\draw[] (1,4) -- (3,4);
\draw[] (2,4) -- (3,3);
\draw[] (1,4) -- (2,8);
\draw[] (2,8) -- (3,7);
\draw[] (1,8) -- (3,8);
\draw[] (1,8) -- (2,4);

\node[anchor=south] at (0,1) {1};
\node[anchor=south] at (2,1) {2};
\node[anchor=south] at (3,1) {3};
\node[anchor=south] at (4,1) {4};
\node[anchor=south] at (4+0.5*\la,3) {5};
\node[anchor=south] at (3,2) {6};
\node[anchor=south] at (4,2) {7};
\node[anchor=south] at (4,4) {8};
\node[anchor=south] at (2+0.5*\la,5) {9};
\node[anchor=south] at (3,5) {10};
\node[anchor=south] at (4,5) {11};
\node[anchor=south] at (4+\la,7) {12};
\node[anchor=south] at (3,6) {13};
\node[anchor=south] at (4,6) {14};
\node[anchor=south] at (4+\la,8) {15};
\node[anchor=south] at (5*\la,5) {16};
\node[anchor=south] at (1+3*\la,4) {17};
\node[anchor=south] at (2,4) {18};
\node[anchor=south] at (3-\la,3) {19};
\node[anchor=south] at (3-\la,4) {20};
\node[anchor=south] at (2-3.5*\la,8) {21};
\node[anchor=south] at (3-1.25*\la,7) {22};
\node[anchor=south] at (3-1.23*\la,8) {23};
\node[anchor=south] at (1+3.75*\la,8) {24};

\foreach \x in {1,...,8} {
\node[anchor=west] at (4,\x) {\small $y_\x$};
}

\node[anchor=north] at (3-3*\la,1) {\small $y_1,y_3$};
\node[anchor=north] at (3-3*\la,2) {\small $y_2,y_4$};
\node[anchor=north] at (2,1) {\small $y_1^4$};
\node[anchor=north] at (-\la,1) {\small $y_1^8$};
\node[anchor=north] at (2,5) {\small $y_5^8$};
\node[anchor=north] at (3-3.25*\la,5) {\small $y_5,y_7$};
\node[anchor=north] at (3-3*\la,6) {\small $y_6,y_8$};
\node[anchor=north] at (3,3) {\small $\triangleright$};
\node[anchor=north] at (3+5*\la,4) {\small $y_2,y_4,s^{(1)}_1$};
\node[anchor=north] at (2+\la,4) {\small $\triangleleft$};
\node[anchor=north] at (1-5.75*\la,4) {\small $y_1^8,\hat{t}_1^{(1)2}$};
\node[anchor=north] at (1-5*\la,8) {\small $y_1^8,\hat{t}_1^{(1)2},\hat{t}_1^{(2)}$};
\node[anchor=north] at (3,7) {\small $\diamond$};
\node[anchor=north] at (3+5*\la,8) {\small $y_6,y_8,s_1^{(2)}$};
\node[anchor=north] at (2,8) {\small $y_5^8,s_1^{(2)2},\hat{t}_1^{(1)2}$};


\node[] at (0,6.1) {\small $\triangleright = y_1,y_3,\hat{t}^{(1)2}_1,s^{(1)}_1$};
\node[] at (-0.1,6.5) {\small $\triangleleft = y_1^4,s^{(1)2}_1,\hat{t}^{(1)2}_1$};
\node[] at (-0.05,6.9) {\small $\diamond = y_5,y_7,s_1^{(2)},\hat{t}_2^{(1)}$};

%
%

\end{tikzpicture}}
\caption{\small Logical structure of the successive cancellation decoder for the setup $L=4$, $M=2$, $\mathcal{I}=\{2,3\}$, $K=2$ and $\mathcal{O}=\{2\}$ (compare with~\cite[Fig.\ 10]{arikan09}). Note that $\hat{t}_j^{(i)}$ denotes the $j$th output of the $i$th decompressor at the outer layer and $s_j^{(i)}$ denotes the $j$th frozen bit at the $i$th inner encoding block. The numbering of the nodes represents the order in which they get activated in the decoding process.}
\label{fig:clDecoder}
\end{figure}
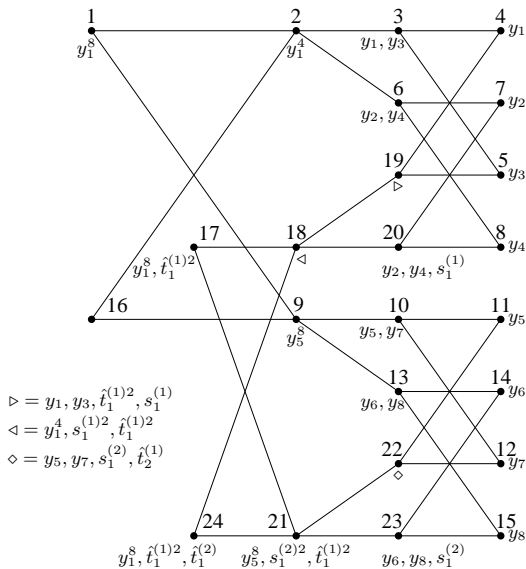

\begin{mylem} \label{lem:ApDecoder}
The classical decoder explained above has complexity $O(N\log N)$.  
\end{mylem}

\begin{IEEEproof}
Let $V_i$ denote the $i$th non-frozen output of an inner encoding block. The decoder proceeds by employing, in sequence, the $K$ decompressors for blocklength-$M$ compression of $V_i$ given $Y^LV^{i-1}$. This ensures that at all times the decoder has all the required previous inputs $V^{i-1}$. 
Each decompressor can be executed using $O(M\log M)$ operations, given the corresponding likelihood ratio of $V_i|Y^LV^{i-1}$. All such likelihoods can be computed in $O(L\log L)$ steps, and each of the $M$ inner encoding blocks requires its own likelihood calculation, as the values taken by $V^{i-1}$ can differ in each case. Using $K=O(L)$ which follows from the polarization phenomenon (cf.\ Theorem~\ref{thm:polarizationPhenomenon}), we find that the decoder has complexity $O(N\log N)$.  
\end{IEEEproof}

We next analyze the reliability of the multilevel encoder and decoder explained above. Suppose we
would like to compress ($L$ instances of) $(V_1,\dots,V_n)$ relative
to side information $Y$, by sequentially compressing $V_i$ relative to
$V^{i-1}Y$. Define $\hat{V}_i$ to be the output of the decompressor,
let $\mathcal{A}_i$ be the event that $\hat{V}_i \neq V_i$ (i.e.,\ that
the decompressor makes a mistake at position $i$), and let
$\mathcal{B}_i:= \cup_{k=1}^i \mathcal{A}_k$. Note that
$\Prv{\mathcal{B}_n}$ is the probability of incorrectly decoding at
least one $V_i$ for $i \in \left[n\right]$. Let $r$ be a bound on the
probability of that we decode incorrectly at any step and that the
previous steps are all correct: $\Prv{\mathcal{A}_j \cap
  \mathcal{B}^c_{j-1}}\leq r$ for all $j\in[n]$. Then
\begin{mylem} \label{lem: newerrLem}
For $n\in \mathbb{Z}^{+}$ and $r$ as defined above, we have
 \begin{equation}
  \Prv{\mathcal{B}_n} \leq nr.
 \end{equation}
\end{mylem}
\begin{IEEEproof}
 The proof proceeds by induction over $n$; the case $n=1$ holds by assumption. The induction step is as follows:
\begin{align}
 \Prv{\mathcal{B}_{n+1}} &= \Prv{\mathcal{B}_n \cup \mathcal{A}_{n+1}}\\
 &= \Prv{\mathcal{B}_n} + \Prv{\mathcal{A}_{n+1} \cap \mathcal{B}_n^c}\\
 &\leq \Prv{\mathcal{B}_n} + r \label{eq: comment1}\\
 &\leq (n+1)r. \label{eq: comment2}
\end{align}
where \eqref{eq: comment1} follows by assumption and \eqref{eq: comment2} uses the induction hypothesis.
\end{IEEEproof}

Now the reliability statement follows easily. 

\begin{mylem}
\label{lem:reliabilityCSC}
The error probability of the encoder and decoder introduced above satisfies $P_{\rm err}=O(L\,2^{-M^\beta})$ for any $\beta<\frac12$.
 \end{mylem}

\begin{IEEEproof}
	For the polar source coding scheme, note that $\Prv{\mathcal{A}_i \cap \mathcal{B}_{i-1}^c} + x \in O(2^{-M^\beta})$, where $x$ is the probability that $\hat{V}_i \neq V_i$ given that a mistake previously occurred, but where we still give the correct $V^{i-1}$ to the decompressor. We can therefore upper bound $r$ in Lemma~\ref{lem: newerrLem} by $O(2^{-M^\beta})$~\cite{arikan10}.   Thus, the probability of incorrectly decoding any of the $K$ $V_i$ is $O(L 2^{-M^\beta})$.
\end{IEEEproof}

We note that for the phase reconciliation task we cannot directly use that decoder since we have measured some qubits at the inner layer with respect to the amplitude basis (the qubits belonging to the system $\bar A^{\setC}$), which implies that we do not have knowledge about the phase-basis-measurements of these qubits. In the following we present two different approaches to resolve this problem. We first show that the phase-basis measurements we do not know can be chosen at random without affecting the decoder's reliability. Alternatively, we show how to adapt the classical decoding algorithm introduced in \cite{suttershort} such that it does not require the above mentioned phase-basis measurements and still remains computationally efficient. 

\prlsection{Choosing Measured Qubits at Inner Layer at Random} \label{ap:randomchoose}
We show that randomly choosing the qubits that have been measured with respect to the amplitude basis at the inner layer does not affect the reliability of the outer IR decoder for symmetric sources. This implies that the computationally efficient decoder introduced in \cite{suttershort} and explained above can be used to reliably decode the phase IR.

Recall that the state
\begin{equation}
\ket{\psi}^{ABCR} = \sum_{z} \sqrt{p_z} \ket{z}^A \ket{z}^{C} \ket{\varphi_z}^{BR}
\end{equation}
produced by the source is relevant for the phase IR task. Taking $L$ copies and applying the polarization transform $G_L$ to $A^L$ and $C^L$ gives
\begin{align}
&\ket{\Psi}^{A^L B^L C^L R^L} \nonumber \\
&= \sum_{z^L \in \{0,1 \}^L} \sqrt{p_{z^L}} \ket{G_L z^L}^{A^L} \ket{G_L z^L}^{C^L} \ket{\varphi_{z^L}}^{B^L R^L}. 
\end{align}
Expressing $A^L$ in the conjugate phase basis gives
\begin{align}
&\ket{\Psi}^{A^L B^L C^L R^L} \nonumber \\
&\hspace{5mm}= \frac{1}{\sqrt{2^L}} \sum_{x^L,z^L \in \{0,1 \}^L} \sqrt{p_{z^L}} (-1)^{x^L \cdot G_L z^L} \ket{\tilde x^L}^{A^L} \nonumber \\
&\hspace{43mm} \ket{G_L z^L}^{C^L} \ket{\varphi_{z^L}}^{B^L R^L}\\
&\hspace{5mm}=\frac{1}{\sqrt{2^L}} \sum_{x^L \in \{0,1 \}^L} \ket{\tilde x^L}^{A^L} \left(Z^{x^L} \right)^{C^L} \ket{\xi}^{C^L B^L R^L},
\end{align}
where 
\begin{equation}
\ket{\xi}^{C^L B^L R^L}\!=\! \!\!\!\sum_{z^L \in \{0,1\}^L} \!\!\! \sqrt{p_{z^L}} \ket{G_L z^L}^{C^L} \ket{\varphi_{z^L}}^{B^L R^L}.
\end{equation}
The state of $B^L C^L$ conditioned on the value $x^L$ on $A^L$ is just $(Z^{x^L})^{C^L} \xi^{C^L B^L} (Z^{x^L})^{C^L}$. 
The source is symmetric because the conditional states are all related by unitary action, here on the $C^L$ systems.

 Dividing $A^L$ and $C^L$ into two systems $\bar A$ and $\bar A^{\setC}$ respectively $\bar C$ and $\bar C^{\setC}$ enables us to write the marginal state of $B^L C^L$ conditioned on only $\bar x$ as
\begin{equation}
\Theta_{\bar x}^{B^L C^L}= \frac{1}{2^{\left | \bar A^{\setC} \right|}}\!\!\! \sum_{\bar x^{\setC} \in \{0,1\}^{\left| \bar A^{\setC} \right|}}\!\!\! \left(Z^{\bar x}\right)^{\bar C}\!\! \left(Z^{\bar x^{\setC}}\right)^{\bar C^{\setC}}\!\! \xi^{C}\!\!  \left(Z^{\bar x}\right)^{\bar C}\! \left(Z^{\bar x^{\setC}}\right)^{\bar C^{\setC}}\!\!.
\end{equation}
Recall that we would like to use the recursive likelihood formulas in the successive cancellation decoder to determine the likelihoods of $\bar x_i$ given $\bar x^{i-1} B^L C^L$. Note that for this likelihoods only the state $\ket{\Psi}^{\bar A B^L C^L}$ is relevant---$\bar A^{\setC}$ has been discarded.

We can mimic this having this state in the standard polar coding setup as follows. Given $\ket{\Psi}^{A^L B^L C^L R^L}$, applying a CPTP map to $C^L$ that randomly performs a $Z$ operation on each of the qubits in $\bar A^{\setC}$ yields a state whose $B^L C^L$ marginal given $X^{\bar A}$ is the same as $\Theta_{\bar x}^{B^L C^L}$. Moreover, the bits $\bar x^{\setC}$ are now uncorrelated with $B^L C^L$ and $\bar x$, which can be seen from direct calculation:
$Z^{\bar x^{\setC}}$ acting on $\bar C^{\setC}$ goes to $Z^{u+\bar x^{\setC}}$ for a random $u$, so the $\bar x^{\setC}$ dependence is eliminated.

 Now proceed with the usual recursive likelihood calculation. Nominally, these likelihood (ratios) are functions of the bits $\bar x^{\setC}$ . However, since these are independent of everything else, the ratios computed by the recursion formulas using $\bar x^{\setC}$ have the same values as the ratios we are looking for. The dependence on on $\bar x^{\setC}$ is irrelevant. Moreover, the decoder may as well choose $\bar x^{\setC}$ himself.
 

\prlsection{Neglecting Measured Qubits at Inner Layer} \label{ap:neglect}
We show how to efficiently compute the likelihood ratios at the outer layer without using the qubits measured with respect to the amplitude basis at the inner layer. Let $\mathcal{F}\subseteq [N]$ denote the indices of the qubits measured at the inner layer, i.e.,\ the qubits belonging to the $M$ systems $\bar A^{\setC}$ and let $\mathcal{F}^{\setC}:=[N]\backslash \mathcal{F}$. For each $iÊ\subseteq [N]$ we define $\mathcal{F}_{(i)}:=[i]\cap \mathcal{F}$ and $\mathcal{F}_{(i)}^{\setC}:=[i] \cap \mathcal{F}^{\setC}$. 

In order to execute Ar{\i}kan's classical SC-decoder at the outer layer, we need to compute for each $i \in [N]$ such that $X^{A^N}_i \in \mathcal{F}^{\setC}$ the following likelihood ratio
\begin{align}
&L^{(i)}(b^N,\hat x[\mathcal{F}_{(i)}^{\setC}]):= \nonumber \\
&\hspace{3mm}\frac{\Prvcond{X^{A^N}_i=0}{B^N=b^N, \bigcap \limits_{j \in \mathcal{F}_{(i)}^{\setC}} \left \lbrace X^{A^N}_j = \hat x_j \right \rbrace}}{\Prvcond{X^{A^N}_i=1}{B^N=b^N, \bigcap \limits_{j \in \mathcal{F}_{(i)}^{\setC}} \left \lbrace X^{A^N}_j = \hat x_j \right \rbrace}}. \label{eq:LRcomplex}
\end{align}
However, since in \eqref{eq:LRcomplex} the qubits measured with respect to the amplitude basis at the inner layer are missing in the conditioning, it is not straightforward how to compute this likelihood ratio efficiently.\footnote{Ar{\i}kan's recursive formula cannot be applied directly.} Using Bayes' theorem we can write
\begin{align}
&\Prvcond{X^{A^N}_i=0}{B^N=b^N, \bigcap \limits_{j \in \mathcal{F}_{(i)}^{\setC}} \left \lbrace X^{A^N}_j = \hat x_j \right \rbrace} \nonumber \\
&= \sum_{\hat x_k \in \{0,1\}\, \forall k \in \mathcal{F}_{(i)}} \mathrm{Pr}\left[X^{A^N}_i=0, \bigcap \limits_{k \in \mathcal{F}_{(i)}} \left \lbrace X^{A^N}_k=\hat x_k \right \rbrace \bigg| \right. \nonumber \\
&\hspace{25mm} \left. B^N=b^N, \bigcap \limits_{j \in \mathcal{F}_{(i)}^{\setC}} \left \lbrace X^{A^N}_j = \hat x_j \right \rbrace \right]\\
&= \left( \frac{1}{2} \right)^{\left| \mathcal{F}_{(i)} \right|} \!\!\!\!\!\!\!\!\!\!\!\! \sum_{\hat x_k \in \{0,1\}\, \forall k \in \mathcal{F}_{(i)}} \!\!\!\!\!\!\! \mathrm{Pr}\left[X^{A^N}_i=0, \bigcap \limits_{k \in \mathcal{F}_{(i)}} \left \lbrace X^{A^N}_k=\hat x_k \right \rbrace \big|  \right. \nonumber \\
&\hspace{8mm} \left. B^N=b^N, \bigcap \limits_{j \in \mathcal{F}_{(i)}^{\setC}} \left \lbrace X^{A^N}_j = \hat x_j \right \rbrace \right] /  \nonumber \\
&\hspace{35mm}\Prv{ \bigcap \limits_{k \in \mathcal{F}_{(i)}} \left \lbrace X^{A^N}_k=\hat x_k \right \rbrace} \label{eq:stepunif}\\
&= \left( \frac{1}{2} \right)^{\left| \mathcal{F}_{(i)} \right|} \sum_{\hat x_k \in \{0,1\}\, \forall k \in \mathcal{F}_{(i)}} \mathrm{Pr}\left[X^{A^N}_i=0\big| B^N=b^N, \right. \nonumber \\
&\hspace{45mm} \left.\bigcap \limits_{j=1}^{i-1} \left \lbrace X^{A^N}_j = \hat x_j \right \rbrace \right], \label{eq:LRArikanForm}
\end{align}
where \eqref{eq:stepunif} uses that the random variables $X_j^{A^N}$ for $j \in \mathcal{F}_{j}$ are independent uniformly distributed, since they have been measured in the complementary amplitude basis. The elements of the final sum \eqref{eq:LRArikanForm} can be computed efficiently using Ar{\i}kan's recursive formula \cite{arikan09}. Since the elements of the sum in \eqref{eq:LRArikanForm} are bounded, we can sample the sum obtaining a (reasonably) good approximation of the true value which can be done efficiently.

We are now ready to state the main result of this section---that the distillation scheme introduced in Section~\ref{sec:ed} for Pauli or erasure channels is efficiently decodable when using quantum polar codes.
\begin{mythm} \label{thm:Bob}
The decoding of the distillation scheme can be done with $O(N \log N)$ steps.
\end{mythm}
\begin{IEEEproof}
According to \cite{arikan09}, the $M$ amplitude recovery blocks together have a complexity of $O(M L \log L)$. Using Lemma~\ref{lem:ApDecoder} it follows that the phase correction can be done in $O(N \log N)$. Since $N=ML$, this proves the assertion.
\end{IEEEproof}

When using quantum polar codes for Pauli or erasure channels we can derive explicit expressions for $\epsilon_1$ and $\epsilon_2$ and hence make a precise statement about the reliability of the distillation scheme.
\begin{mycor} \label{cor:quality}
The reliability of the distillation scheme scheme is as given in Proposition~\ref{prop:quality} with $\epsilon_1=O(2^{-L^{\beta}})$ and $\epsilon_2=O(L 2^{-M^{\beta'}})$ for any $\beta,\beta'<\tfrac{1}{2}$.
\end{mycor}
\begin{IEEEproof}
Let $p_{\rm err}(\mathcal{D}_{A})$ denote the error probability of the decoding operation $\mathcal{D}_A$ and $p_{\rm err}(\mathcal{D}_{P})$ the error probability of the decoding process $\mathcal{D}_{P}$. The error probability for all decoding operations with respect to the amplitude basis is denoted by $p_{\rm err}(\mathcal{D}_{A^M})$. According to Proposition~\ref{prop:quality}, the trace distance between the scheme's output and a maximally entangled state of appropriate dimension is less than $\sqrt{2 p_{\rm err}(\mathcal{D}_{A^M}) } + \sqrt{2 p_{\rm err}(\mathcal{D}_P)}$. Using the union bound we obtain $p_{\rm err}(\mathcal{D}_{A^M}) \leq M p_{\rm err}(\mathcal{D}_{A})$. Since we use the standard polar decoder \cite{arikan09} for the amplitude error-correction, we have 
\begin{equation}
p_{\rm err}\left(\mathcal{D}_{A}\right)=O\left(2^{-L^{\beta}}\right) \quad \textnormal{for} \quad \beta<\frac{1}{2}.
\end{equation}
Furthermore, according to Lemma~\ref{lem:reliabilityCSC},
\begin{equation}
p_{\rm err}\left(\mathcal{D}_{P}\right)=O\left(L 2^{-M^{\beta'}}\right) \textnormal{ for any }\beta'<\frac{1}{2},
\end{equation}
which proves the assertion.
\end{IEEEproof}


\section{Channel Coding} \label{sec:channelcoding}
Bennett \emph{et al.}\ \cite{bennett_mixed-state_1996} showed that any entanglement distillation scheme can be turned into a channel coding scheme, which however is not known to be computationally efficient---even if the entanglement distillation protocol we started with is efficient. In this section, we show how to modify the entanglement distillation scheme introduced in Section~\ref{sec:ed} such that it can be used for efficient channel coding. The resulting coding scheme is depicted schematically in Figure~\ref{fig:channelview}. 
Before applying the actual encoding transformation, the outer encoder adds redundancy in form of random qubits which are sent to the decoder. As explained in the entanglement distillation scheme in the previous section, we know that after the inner layer the state is perfectly known with respect to the amplitude basis. Therefore, we can choose the additional qubits at random in the complementary phase basis.

The inner encoder also adds redundancy. The additional qubits are generated as explained in \cite[Section II]{suttershort} and sent to the decoder, before applying the actual encoding transform. The encoded data is then transmitted over $N$ identical channels $\mathcal{N}$. The decoding is identical to Bob's task in the entanglement distillation scenario, explained in Section~\ref{sec:ed}.

For the code construction, the set of frozen qubits (the indices which determine at which position the redundant qubits are added) at the outer and inner layer have to be determined. This can be done efficiently for the inner layer as explained in Section~\ref{sec:effQPC}. The existence of an efficient algorithm for the outer layer remains an open question. 
\begin{figure}[!htb]
\hspace{-21mm}
\scalebox{0.8}{
\def \xcomp{1.0}
\def \ycomp{5}
\def \xblock{2.93}
\def \yblock{3}

\def \yPA{1.75}

\def \xgapd{1.5} 
\def \ygapd{0.9} 

\def \ybs{-0.3} 

\def \ys{0.1} 

\def \xdec{0.9} 
\def \w{0.526} 

\def \xs{0.4} 

\def \s{2.8} 

\def \c{0.4} 

\def \a{0.1} 

\def \la{0.2} 

\def \xdec{1}

\def \xdash{0.2} 
\def \ydash{0.2} 

\def \xstart{0.8}

\begin{tikzpicture}[scale=1,auto, node distance=1cm,>=latex']
	
     \draw [draw] (\xgapd+\xcomp,-2.5*\w-2*\ys+0.5*\w+0.5*\ys) -- (2*\xcomp+\xgapd,-2.5*\w-2*\ys+0.5*\w+0.5*\ys); 
     \draw [draw] (\xgapd+\xcomp,-2.5*\w-2*\ys+0.5*\w+0.5*\ys) -- (\xgapd+\xcomp,-\yPA-2*\ybs-\s);
     \draw [draw] (\xgapd+\xcomp,-\yPA-2*\ybs-\s) -- (2*\xcomp+\xgapd,-\yPA-2*\ybs-\s);
     \draw [draw] (2*\xcomp+\xgapd,-2.5*\w-2*\ys+0.5*\w+0.5*\ys) -- (2*\xcomp+\xgapd,-\yPA-2*\ybs-\s);     
    \node at (1.5*\xcomp+\xgapd,-1.32*\w-0.99*\ys-0.33*\yPA-0.66*\ybs-0.33*\s) {outer};
    \node at (1.5*\xcomp+\xgapd,-0.66*\w-0.495*\ys-0.66*\yPA-1.32*\ybs-0.66*\s) {$\enc$};

  \draw [draw] (2*\xcomp+2*\xgapd,0) -- (2*\xcomp+2*\xgapd+\xcomp,0);  
  \draw [draw] (2*\xcomp+2*\xgapd,-\yblock-2*\ybs) -- (2*\xcomp+2*\xgapd+\xcomp,-\yblock-2*\ybs);
  \draw [draw] (2*\xcomp+2*\xgapd,0) -- (2*\xcomp+2*\xgapd,-\yblock-2*\ybs);  
 \draw [draw] (2*\xcomp+2*\xgapd+\xcomp,-\yblock-2*\ybs) -- (2*\xcomp+2*\xgapd+\xcomp,0);
     \node [] at (2*\xcomp+2*\xgapd+0.5*\xcomp,-0.33*\yblock-0.66*\ybs) {inner};
     \node [] at (2*\xcomp+2*\xgapd+0.5*\xcomp,-0.66*\yblock-1.32*\ybs) {$\enc$};
     
    \draw [draw] (3*\xcomp+3*\xgapd,0) -- (3*\xcomp+3*\xgapd+\w,0);  
    \draw [draw] (3*\xcomp+3*\xgapd,-\w) -- (3*\xcomp+3*\xgapd+\w,-\w);
    \draw [draw] (3*\xcomp+3*\xgapd,0) -- (3*\xcomp+3*\xgapd,-\w);
    \draw [draw] (3*\xcomp+3*\xgapd+\w,0) -- (3*\xcomp+3*\xgapd+\w,-\w);
   \node at (3*\xcomp+3*\xgapd+0.5*\w,-0.5*\w) {$\mathcal{N}$};
   
    \draw [draw] (3*\xcomp+3*\xgapd,-\ys-\w) -- (3*\xcomp+3*\xgapd+\w,-\ys-\w);  
    \draw [draw] (3*\xcomp+3*\xgapd,-\w-\ys-\w) -- (3*\xcomp+3*\xgapd+\w,-2*\w-\ys);
    \draw [draw] (3*\xcomp+3*\xgapd,-\ys-\w) -- (3*\xcomp+3*\xgapd,-2*\w-\ys);
    \draw [draw] (3*\xcomp+3*\xgapd+\w,-\ys-\w) -- (3*\xcomp+3*\xgapd+\w,-2*\w-\ys);
   \node at (3*\xcomp+3*\xgapd+0.5*\w,-1.5*\w-\ys) {$\mathcal{N}$};
   
    \draw [draw] (3*\xcomp+3*\xgapd,-2*\ys-2*\w) -- (3*\xcomp+3*\xgapd+\w,-2*\ys-2*\w);  
    \draw [draw] (3*\xcomp+3*\xgapd,-3*\w-2*\ys) -- (3*\xcomp+3*\xgapd+\w,-3*\w-2*\ys);
    \draw [draw] (3*\xcomp+3*\xgapd,-2*\ys-2*\w) -- (3*\xcomp+3*\xgapd,-3*\w-2*\ys);
    \draw [draw] (3*\xcomp+3*\xgapd+\w,-2*\ys-2*\w) -- (3*\xcomp+3*\xgapd+\w,-3*\w-2*\ys);
   \node at (3*\xcomp+3*\xgapd+0.5*\w,-2.5*\w-2*\ys) {$\mathcal{N}$};
   
    \draw [draw] (3*\xcomp+3*\xgapd,-3*\ys-3*\w) -- (3*\xcomp+3*\xgapd+\w,-3*\ys-3*\w);  
    \draw [draw] (3*\xcomp+3*\xgapd,-4*\w-3*\ys) -- (3*\xcomp+3*\xgapd+\w,-4*\w-3*\ys);
    \draw [draw] (3*\xcomp+3*\xgapd,-3*\ys-3*\w) -- (3*\xcomp+3*\xgapd,-4*\w-3*\ys);
    \draw [draw] (3*\xcomp+3*\xgapd+\w,-3*\ys-3*\w) -- (3*\xcomp+3*\xgapd+\w,-4*\w-3*\ys);
   \node at (3*\xcomp+3*\xgapd+0.5*\w,-3.5*\w-3*\ys) {$\mathcal{N}$};

  \draw [draw] (2*\xcomp+2*\xgapd,-\s) -- (2*\xcomp+2*\xgapd+\xcomp,-\s);  
  \draw [draw] (2*\xcomp+2*\xgapd,-\yblock-2*\ybs-\s) -- (2*\xcomp+2*\xgapd+\xcomp,-\yblock-2*\ybs-\s);
  \draw [draw] (2*\xcomp+2*\xgapd,-\s) -- (2*\xcomp+2*\xgapd,-\yblock-2*\ybs-\s);  
 \draw [draw] (2*\xcomp+2*\xgapd+\xcomp,-\yblock-2*\ybs-\s) -- (2*\xcomp+2*\xgapd+\xcomp,-\s);
     \node [] at (2*\xcomp+2*\xgapd+0.5*\xcomp,-0.33*\yblock-0.66*\ybs-\s) {inner};
     \node [] at (2*\xcomp+2*\xgapd+0.5*\xcomp,-0.66*\yblock-1.32*\ybs-\s) {$\enc$};
          
    \draw [draw] (3*\xcomp+3*\xgapd,-\s) -- (3*\xcomp+3*\xgapd+\w,-\s);  
    \draw [draw] (3*\xcomp+3*\xgapd,-\w-\s) -- (3*\xcomp+3*\xgapd+\w,-\w-\s);
    \draw [draw] (3*\xcomp+3*\xgapd,-\s) -- (3*\xcomp+3*\xgapd,-\w-\s);
    \draw [draw] (3*\xcomp+3*\xgapd+\w,-\s) -- (3*\xcomp+3*\xgapd+\w,-\w-\s);
   \node at (3*\xcomp+3*\xgapd+0.5*\w,-0.5*\w-\s) {$\mathcal{N}$};
   
    \draw [draw] (3*\xcomp+3*\xgapd,-\ys-\w-\s) -- (3*\xcomp+3*\xgapd+\w,-\ys-\w-\s);  
    \draw [draw] (3*\xcomp+3*\xgapd,-\w-\ys-\w-\s) -- (3*\xcomp+3*\xgapd+\w,-2*\w-\ys-\s);
    \draw [draw] (3*\xcomp+3*\xgapd,-\ys-\w-\s) -- (3*\xcomp+3*\xgapd,-2*\w-\ys-\s);
    \draw [draw] (3*\xcomp+3*\xgapd+\w,-\ys-\w-\s) -- (3*\xcomp+3*\xgapd+\w,-2*\w-\ys-\s);
   \node at (3*\xcomp+3*\xgapd+0.5*\w,-1.5*\w-\ys-\s) {$\mathcal{N}$};
   
    \draw [draw] (3*\xcomp+3*\xgapd,-2*\ys-2*\w-\s) -- (3*\xcomp+3*\xgapd+\w,-2*\ys-2*\w-\s);  
    \draw [draw] (3*\xcomp+3*\xgapd,-3*\w-2*\ys-\s) -- (3*\xcomp+3*\xgapd+\w,-3*\w-2*\ys-\s);
    \draw [draw] (3*\xcomp+3*\xgapd,-2*\ys-2*\w-\s) -- (3*\xcomp+3*\xgapd,-3*\w-2*\ys-\s);
    \draw [draw] (3*\xcomp+3*\xgapd+\w,-2*\ys-2*\w-\s) -- (3*\xcomp+3*\xgapd+\w,-3*\w-2*\ys-\s);
   \node at (3*\xcomp+3*\xgapd+0.5*\w,-2.5*\w-2*\ys-\s) {$\mathcal{N}$};
   
    \draw [draw] (3*\xcomp+3*\xgapd,-3*\ys-3*\w-\s) -- (3*\xcomp+3*\xgapd+\w,-3*\ys-3*\w-\s);  
    \draw [draw] (3*\xcomp+3*\xgapd,-4*\w-3*\ys-\s) -- (3*\xcomp+3*\xgapd+\w,-4*\w-3*\ys-\s);
    \draw [draw] (3*\xcomp+3*\xgapd,-3*\ys-3*\w-\s) -- (3*\xcomp+3*\xgapd,-4*\w-3*\ys-\s);
    \draw [draw] (3*\xcomp+3*\xgapd+\w,-3*\ys-3*\w-\s) -- (3*\xcomp+3*\xgapd+\w,-4*\w-3*\ys-\s);
   \node at (3*\xcomp+3*\xgapd+0.5*\w,-3.5*\w-3*\ys-\s) {$\mathcal{N}$};

    \draw [draw] (3*\xcomp+4*\xgapd+\w,0) -- (3*\xcomp+4*\xgapd+\w+\xdec,0); 
    \draw [draw] (3*\xcomp+4*\xgapd+\w,-\yblock-2*\ybs-\s) -- (3*\xcomp+4*\xgapd+\w+\xdec,-\yblock-2*\ybs-\s); 
    \draw [draw] (3*\xcomp+4*\xgapd+\w,0) -- (3*\xcomp+4*\xgapd+\w,-\yblock-2*\ybs-\s);  
    \draw [draw] (3*\xcomp+4*\xgapd+\w+\xdec,-\yblock-2*\ybs-\s) -- (3*\xcomp+4*\xgapd+\w+\xdec,0);      
    \node at (3*\xcomp+4*\xgapd+\w+0.5*\xdec,-0.5*\yblock-1*\ybs-0.5*\s) {$\dec$};

            
    \draw [->] (2*\xcomp+\xgapd,-2.5*\w-2*\ys) -- (2*\xcomp+2*\xgapd,-2.5*\w-2*\ys);
    \draw [->] (2*\xcomp+\xgapd,-3.5*\w-3*\ys) -- (2*\xcomp+2*\xgapd,-3.5*\w-3*\ys);        
  

         \draw [->] (\xcomp+\xgapd-\xstart,-2*\w-1.5*\ys-0.5*\s) -- (\xcomp+\xgapd,-2*\w-1.5*\ys-0.5*\s);    
     

    \draw [->] (2*\xcomp+\xgapd,-0.5*\w-\s) -- (2*\xcomp+2*\xgapd,-0.5*\w-\s); 
    \draw [->] (2*\xcomp+\xgapd,-1.5*\w-\ys-\s) -- (2*\xcomp+2*\xgapd,-1.5*\w-\ys-\s);     

    \draw [->] (3*\xcomp+2*\xgapd,-0.5*\w) -- (3*\xcomp+3*\xgapd,-0.5*\w); 
    \draw [->] (3*\xcomp+3*\xgapd +\w,-0.5*\w) -- (3*\xcomp+4*\xgapd+\w,-0.5*\w);
    
    \draw [->] (3*\xcomp+2*\xgapd,-1.5*\w-\ys) -- (3*\xcomp+3*\xgapd,-1.5*\w-\ys); 
    \draw [->] (3*\xcomp+3*\xgapd +\w,-1.5*\w-\ys) -- (3*\xcomp+4*\xgapd+\w,-1.5*\w-\ys);
        
    \draw [->] (3*\xcomp+2*\xgapd,-2.5*\w-2*\ys) -- (3*\xcomp+3*\xgapd,-2.5*\w-2*\ys); 
    \draw [->] (3*\xcomp+3*\xgapd+\w,-2.5*\w-2*\ys) -- (3*\xcomp+4*\xgapd+\w,-2.5*\w-2*\ys);
     
    \draw [->] (3*\xcomp+2*\xgapd,-3.5*\w-3*\ys) -- (3*\xcomp+3*\xgapd,-3.5*\w-3*\ys);             
    \draw [->] (3*\xcomp+3*\xgapd+\w,-3.5*\w-3*\ys) -- (3*\xcomp+4*\xgapd+\w,-3.5*\w-3*\ys);
    \draw [->] (3*\xcomp+2*\xgapd,-0.5*\w-\s) -- (3*\xcomp+3*\xgapd,-0.5*\w-\s); 
    \draw [->] (3*\xcomp+3*\xgapd+\w,-0.5*\w-\s) -- (3*\xcomp+4*\xgapd+\w,-0.5*\w-\s); 
    
    \draw [->] (3*\xcomp+2*\xgapd,-1.5*\w-\ys-\s) -- (3*\xcomp+3*\xgapd,-1.5*\w-\ys-\s); 
    \draw [->] (3*\xcomp+3*\xgapd +\w,-1.5*\w-\ys-\s) -- (3*\xcomp+4*\xgapd+\w,-1.5*\w-\ys-\s);
        
    \draw [->] (3*\xcomp+2*\xgapd,-2.5*\w-2*\ys-\s) -- (3*\xcomp+3*\xgapd,-2.5*\w-2*\ys-\s); 
    \draw [->] (3*\xcomp+3*\xgapd+\w,-2.5*\w-2*\ys-\s) -- (3*\xcomp+4*\xgapd+\w,-2.5*\w-2*\ys-\s);
     
    \draw [->] (3*\xcomp+2*\xgapd,-3.5*\w-3*\ys-\s) -- (3*\xcomp+3*\xgapd,-3.5*\w-3*\ys-\s);             
    \draw [->] (3*\xcomp+3*\xgapd+\w,-3.5*\w-3*\ys-\s) -- (3*\xcomp+4*\xgapd+\w,-3.5*\w-3*\ys-\s);
    
     \draw [->] (3*\xcomp+4*\xgapd+\w+\xdec,-2*\w-1.5*\ys-0.5*\s) -- (3*\xcomp+4*\xgapd+\w+\xdec+\xstart,-2*\w-1.5*\ys-0.5*\s);

      \node at (3*\xcomp+4*\xgapd+\w+\xdec+\xstart+1.5*\la,-2*\w-1.5*\ys-0.5*\s) {$\ket{\hat \phi}$};          
      \node at (\xcomp+\xgapd-\xstart-1.5*\la,-2*\w-1.5*\ys-0.5*\s) {$\ket{\phi}$};

 \draw [] (2.5*\xcomp+2*\xgapd,0) -- (2.5*\xcomp+2*\xgapd,+3*\la); 
 \draw [] (2.5*\xcomp+2*\xgapd,3*\la) -- (3*\xcomp+4*\xgapd+\w+0.5*\xdec,3*\la); 
 \draw [->] (3*\xcomp+4*\xgapd+\w+0.5*\xdec,3*\la) -- (3*\xcomp+4*\xgapd+\w+0.5*\xdec,0); 
      \node at (2.5*\xcomp+2*\xgapd+5*\la,3*\la-1.5*\la) {$B_{C,1}$}; 
 
 \draw [] (2.5*\xcomp+2*\xgapd,-\yblock-2*\ybs-\s) -- (2.5*\xcomp+2*\xgapd,-\yblock-2*\ybs-\s-2*\la); 
 \draw [] (2.5*\xcomp+2*\xgapd,-\yblock-2*\ybs-\s-2*\la) -- (3*\xcomp+4*\xgapd+\w+0.25*\xdec,-\yblock-2*\ybs-\s-2*\la);  
 \draw [->] (3*\xcomp+4*\xgapd+\w+0.25*\xdec,-\yblock-2*\ybs-\s-2*\la) -- (3*\xcomp+4*\xgapd+\w+0.25*\xdec,-\yblock-2*\ybs-\s);  
      \node at (2.5*\xcomp+2*\xgapd+5*\la,-\yblock-2*\ybs-\s-2*\la+\la){$B_{C,2}$};                    

\draw [] (1.5*\xcomp+\xgapd,-\yPA-2*\ybs-\s) -- (1.5*\xcomp+\xgapd,-\yblock-2*\ybs-\s-3*\la);                    
 \draw [] (1.5*\xcomp+\xgapd,-\yblock-2*\ybs-\s-3*\la) -- (3*\xcomp+4*\xgapd+\w+0.75*\xdec,-\yblock-2*\ybs-\s-3*\la);  
  \draw [->] (3*\xcomp+4*\xgapd+\w+0.75*\xdec,-\yblock-2*\ybs-\s-3*\la) -- (3*\xcomp+4*\xgapd+\w+0.75*\xdec,-\yblock-2*\ybs-\s);  
\node at (1.5*\xcomp+\xgapd+3*\la,-\yblock-2*\ybs-\s-3*\la+\la) {$B_{D}$};              
\end{tikzpicture}}
\caption{\small The channel coding view of the scheme for $L=4$ and $M=2$. The outer encoder adds randomly generated qubits in the phase basis, the identities of which are forwarded to the decoder (over a classical channel), before applying the actual encoding transform. At the inner layer the encoder mimics the extra amplitude basis qubits as explained in Section~\ref{ap:clCC} (see also \cite[Section II]{suttershort}) and sends them to the decoder as well. Decoding is the same as  for the entanglement distillation setup explained in Section~\ref{sec:ed}.}
\label{fig:channelview}
\end{figure}
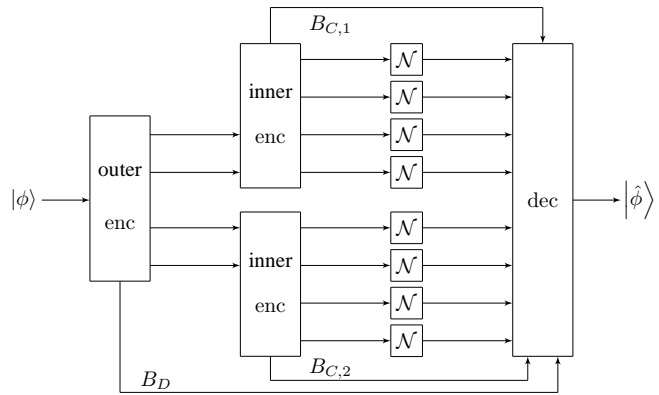

\subsection{Encoding}
We show that for Pauli and erasure channels together with the use of quantum polar codes an efficient encoder and decoder can be obtained. 

\begin{mycor}
For Pauli channels and the use of polar codes, there exists an encoder for the scheme described above that has $O(N \log N)$ complexity.
\end{mycor}
\begin{IEEEproof}
The encoder is equivalent to the one introduced in Section~\ref{sec:effQPC}. Note that at the outer encoder the frozen qubits are chosen at random as justified above. The inner encoder choses its frozen qubits as explained in \cite[Section II]{suttershort}.
\end{IEEEproof}

In order to explain the decoding strategy we first present an efficient concatenated classical channel coding scheme based on polar codes
\subsection{Efficient Concatenated Classical Channel Coding} \label{ap:clCC}
In order to explain the decoding strategy we first present an efficient concatenated classical channel coding scheme based on polar codes. The scheme achieves the capacity and has been introduced in \cite{suttershort}. It will serve as a building block to prove the reliability and efficiency of the quantum channel coding scheme introduced in Section~\ref{sec:channelcoding} when using quantum polar codes for Pauli and erasure channels.

We consider a discrete memoryless channel $\W: \mathcal{X} \to \mathcal{Y}$ with a binary input alphabet $\mathcal{X}=\{0,1 \}$ and an arbitrary output alphabet $\mathcal{Y}$.
The idea for the efficient classical channel coding scheme is to run the encoder of the source coding scheme---introduced in Section~\ref{sec:effQPC}---in reverse and to use the same decoder. One main difficulty that occurs is that we need to simulate the frozen bits at the inner and outer layer. As we show next, by cleverly choosing the frozen bits at the inner layer (we call this in the following \emph{shaping}), we can approximate the source coding case, i.e., the sequence of $N$ i.i.d.\ correlated random variables $(X,Y)^N$ arbitrarily well for large $N$ (cf.\ Lemma~\ref{lem:inputoutput}). This enables us to prove that the efficient decoder introduced in Section~\ref{sec:effQPC} for source coding, can also be used for reliable channel coding. More details about this approach can be found in \cite{suttershort}.

\prlsection{Shaping} \label{ap:shaper}
The idea of how to approximate $(X,Y)^N$ and hence how to choose the frozen bits at the inner encoding block is to run an extractor for the distribution $(X,Y)^N$ we want to approximate in reverse.\footnote{We only need to show how to approximate $(X,Y)^L$ since taking $M$ identical blocks then leads to an approximation of $(X,Y)^N$.} A priori it is not clear this process can be done efficiently, however luckily we show that this is the case for extractors based on the source polarization phenomenon.

 A $K$-bit polarization extractor $\E_{L,K}$ for
$X^L$ simply outputs the $K$ bits of $U^L=X^LG_L$ for which
$H(U_i|U^{i-1})$ are greatest. We denote this (ordered) set of indices
by $\mathcal{E}_K$ and the output of the extractor by
$U^L[\mathcal{E}_K]$. 

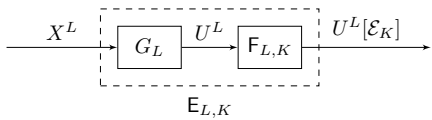
\begin{figure}[!htb]
\hspace{9mm}
\scalebox{0.8}{
\tikzstyle{block1} = [draw, rectangle, 
    minimum height=2em, minimum width=3em, node distance=2.5cm,anchor=center]
\def\gap{1}

\begin{tikzpicture}[auto,>=latex']
	
    \node (input) at (-3.5*\gap,0) {};
    \node [block1] (recon) at (-1*\gap,0) {$G_L$};
    \node [block1] (polar) at (1*\gap,0) {$\F_{L,K}$};
    \node (output) at (3.8*\gap,0) {};

    \draw [->,] (input) -- node {{$X^L$}} (recon); 
    \draw [->,] (recon) -- node {${U}^L$} (polar);
    \draw [->,] (polar) -- node {${U}^L[\mathcal{E}_K]$} (output);
   
    \node[draw,dashed,fit=(recon) (polar),inner sep=8pt] (shaper) {};
    \node[below of=shaper] {$\E_{L,K}$};
    
\end{tikzpicture}}
\caption{\small Polarization-based randomness extractor $\E_{L,K}$. The input $X^L$ is first transformed to $U^L$ via the polarization transformation $G_L$, and subsequently $\F_{L,K}$ filters out the $K$ bits of $U^L$ for which $H(U_i|U^{i-1})$ are greatest.}
\label{fig: Extraction}
\end{figure}
 The aim of randomness extraction is to output $K$ approximately uniform bits, where the approximation is quantified using the variational distance. Recall that for distributions $P$ and $Q$ over the same alphabet $\mathcal{X}$, the variational distance is defined by $\delta(P,Q):=\frac12\sum_{x\in\mathcal{X}}\left|P(x)-Q(x) \right| $. We will often abuse notation slightly and write a random variable instead of its distribution in $\delta$. 
 
Using $\mathcal{E}_K$ we define the shaper for $X^L$ as follows
\begin{mydef}
\label{def: shaper}
For $\mathcal{U}=\{0,1 \}$, the shaper  $\S_{K,L}$ for $X^L$ is the map $\S_{K,L}:\mathcal{U}^K\to\mathcal{X}^L$ taking input $U^K$ to  $\hat{X}^L=\hat{U}^LG_L$, with 
\begin{align}
\label{eq:shaperdef}
\hat{U}_i=\left\{
\begin{array}{ll} U_{{\rm pos}_{\mathcal{E}_K}(i)}& i\in\mathcal{E}_K\\ Z_i & {\rm else}
\end{array}\right. ,
\end{align}
where ${\rm pos}_{\mathcal{A}}(a)$ denotes the position of the entry $a$ in $\mathcal{A}$.
Here $Z_i$ is a random variable generated from the distribution of $U_i|U^{i-1}$, using $U^L=X^LG_L$. 
\end{mydef}

\begin{figure}[!htb]
\hspace{9mm}
\scalebox{0.8}{
\tikzstyle{block1} = [draw, rectangle, 
    minimum height=2em, minimum width=3em, node distance=2.5cm,anchor=center]
\def\gap{1}

\begin{tikzpicture}[auto,>=latex']
	
    \node (input) at (-3.8*\gap,0) {};
    \node [block1] (recon) at (-1*\gap,0) {$\R_{K,L}$};
    \node [block1] (polar) at (1*\gap,0) {$G_L$};
    \node (output) at (3.5*\gap,0) {};

    \draw [->] (input) -- node {{${\tilde{U}}^K$}} (recon); 
    \draw [->] (recon) -- node {$\hat{U}^L$} (polar);
    \draw [->] (polar) -- node {$\hat{X}^L$} (output);
   
    \node[draw,dashed,fit=(recon) (polar),inner sep=8pt] (shaper) {};
    \node[below of=shaper] {$\S_{K,L}$};
   
\end{tikzpicture}}
\caption{\small Generation of an approximation to $X^L$ from a uniform input $\tilde{U}^K$ using the shaper $\S_{K,L}$. $\hat{U}^L$ is first constructed by $\R_{K,L}$ from the uniform input according to (\ref{eq:shaperdef}). Applying $G_L$ gives $\hat{X}^L$, which has nearly the same distribution as $X^L$. }
\label{fig: ExactRec}
\end{figure}
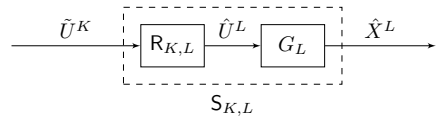

Using the shaper with uniform input $\tilde{U}^K$ (a $K$-bit vector
whose entries are i.i.d.\ $\Bernoulli{\frac12}$) generates an
approximation $\hat{X}^L:=\S_{K,L}(\tilde{U}^K)$ to $X^L$ (see also~\cite[Lemma 11]{korada_polar_2010}).

\begin{mylem}
  \label{lem:shapergood} For $\epsilon \geq 0$ and $K$ such that
  $\Hc{U_i}{U^{i-1}} \geq 1-\epsilon$ for all $i \in
  \mathcal{E}_K$,
  \begin{align}
 \TD{\hat{X}^L}{X^L}\leq K
  \sqrt{\frac{\ln 2}{2} \epsilon} \ .
  \end{align}
\end{mylem}

\begin{IEEEproof}
  Let $\hat{U}^L$ be the $L$-bit string obtained when using the shaper
  with uniform input $\tilde{U}^K$ (cf.\ Equation~\ref{eq:shaperdef}).  We
  have $X^L = U^L G_L$ and $\hat{X}^L = \hat{U}^L G_L$ and, hence,
\begin{align}
  \TD{\hat{X}^L}{X^L} =  \TD{\hat{U}^L}{U^L} \ .
\end{align}
We will bound the distance on the right hand side. For this, we
introduce a family of intermediate distributions $P^{(i)}_{U_1 \cdots
  U_i \hat{U}_{i+1} \cdots \hat{U}_L}$, for $i=0, \ldots, L$, defined
by
\begin{align}
  P^{(i)}_{U_1 \cdots U_i \hat{U}_{i+1} \cdots \hat{U}_L} := P_{U_1
    \cdots U_i } P_{\hat{U}_{i+1} \cdots \hat{U}_L | \hat{U}_1 \cdots
    \hat{U}_i}\ ,
\end{align}
so that $P^{(0)}_{\hat{U}_{1} \cdots \hat{U}_L} = P_{\hat{U}_1 \cdots
  \hat{U}_L}$ and $P^{(L)}_{U_1 \cdots U_L} = P_{U_1 \cdots U_L}$.
 By the triangle inequality,
\begin{align}
 & \TD{\hat{U}^L}{U^L}  \nonumber \\
& \hspace{10mm} \leq
  \sum_{i = 1}^{L}  \TD{P^{(i-1)}_{U_1 \cdots U_{i-1} \hat{U}_{i}\cdots \hat{U}_L}}{P^{(i)}_{U_1 \cdots U_{i} \hat{U}_{i+1}\cdots \hat{U}_L}} \\
  &\hspace{10mm} \leq \sum_{i = 1}^{L}  \TD{P^{(i-1)}_{U_1 \cdots U_{i-1}\hat{U}_{i}}}{P^{(i)}_{U_1 \cdots U_{i-1} U_{i} }} \ , \label{eq_Usum}
\end{align}
where the last line follows from the fact that the variational
distance is non-increasing under stochastic maps~\cite{liesevajda87}
(we apply this to the map that generates $\hat{U}_{i+1} \cdots
\hat{U}_L$ according to the distribution $P_{\hat{U}_{i+1} \cdots
  \hat{U}_L| \hat{U}_{1} \cdots \hat{U}_{i}}$).  Each term of the sum
can be written as $ \delta(P_{U^{i-1}}  P_{\hat{U}_i|\hat{U}^{i-1}}
     , P_{U^{i-1}}  P_{U_i|U^{i-1}})$ or, equivalently, $\mathsf{E}_{U^{i-1}} \left[\delta(P_{\hat{U}_i|\hat{U}^{i-1}}, P_{U_i|U^{i-1} })\right]$. 
     To bound this, we use
Pinsker's inequality~\cite[p.58]{csiszarkorner81} as well as the concavity of the square root,
\begin{align}
 &\mathsf{E}_{U^{i-1}}\left[ \TD{P_{\hat{U}_i|\hat{U}^{i-1}}}{P_{U_i|U^{i-1} }}  \right] \nonumber \\
 &\hspace{15mm}\leq \mathsf{E}_{U^{i-1}}\left[\sqrt{\tfrac{\ln 2}{2}D(P_{U_i|U^{i-1}}\|P_{\hat{U}_i|\hat{U}^{i-1}})}\right]\\
  &\hspace{15mm}\leq \sqrt{\tfrac{\ln 2}{2}\mathsf{E}_{U^{i-1}}\left[D(P_{U_i|U^{i-1}}\|P_{\hat{U}_i|\hat{U}^{i-1}})\right]} \ .
\end{align}
By construction, the conditional
distribution of $\hat{U}_i$ for all $i \in \mathcal{E}_K$ is the uniform
distribution, so that
\begin{align}
 \mathsf{E}_{U^{i-1}}\left[
    D(P_{U_i|U^{i-1}}\|P_{\hat{U}_i|\hat{U}^{i-1}})
 \right]
   &= 1 - \Hc{U_i}{U^{i-1}} \\ 
   &\leq \epsilon \ .   
\end{align}
Furthermore, for all $i \notin \mathcal{E}_K$, the conditional
distribution of $\hat{U}_i$ equals $P_{U_i|U^{i-1}}$, so that the
corresponding term in the sum~\eqref{eq_Usum} vanishes. The sum can
thus be rewritten as
\begin{align}
  \TD{\hat{U}^L}{U^L}
\leq 
  \sum_{i \in \mathcal{E}_K} \sqrt{\tfrac{\ln 2}{2}
    \epsilon} \ ,
\end{align}
from which the assertion follows. 
\end{IEEEproof}

Concatenating the shaper with the channel gives the super-channel
$\W'_{K,L}:=\W^L \circ\S_{K,L}$. Monotonicity of the variational
distance gives the following lemma, which is the basis of our coding
scheme. Letting $\hat{Y}^L:=\W^L(\hat{X}^L)$ and $Y^L =\W^L(X^L)$, we
have
\begin{mylem} 
\label{lem:inputoutput}
For $\epsilon \geq 0$ and $K$ such that
  $\Hc{U_i}{U^{i-1}} \geq 1-\epsilon$ for all $i \in
  \mathcal{E}_K$,
\begin{align}
\TD{\bigl(\tilde{U}^K,\hat{Y}^L\bigr)}{\left(U^L[\mathcal{E}_K],Y^L\right)}\leq
K\sqrt{\frac{\ln 2}{2} \epsilon} \ .
\end{align}
\end{mylem}
\begin{IEEEproof}
  Let $\epsilon'=K \sqrt{\tfrac{\ln 2}{2}\epsilon}$, then Lemma~\ref{lem:shapergood} implies
  $\delta((\hat{X}^L,\hat{Y}^L),(X^L,Y^L))\leq \epsilon'$ by the
  monotonicity of the variational distance under stochastic
  maps. Applying $G_L$ to $X^L$ or $\hat{X}^L$ and marginalizing over
  the elements not in $\mathcal{E}_K$ is also a stochastic map, so
  $\delta((\hat{U}^L[\mathcal{E}_K],\hat{Y}^L),(U^L[\mathcal{E}_K],Y^L))\leq
  \epsilon'$. Observing that $\hat{U}^L[\mathcal{E}_K]=\tilde{U}^K$
  completes the proof.
\end{IEEEproof}

\prlsection{Efficient Classical Encoding and Decoding}
The encoding for channel coding is the reverse operation of the encoding in the source polarization setup explained in Section~\ref{sec:effQPC} and Figure~\ref{fig:clEncoder}. The frozen bits which are sent to the decoder in the source polarization scenario are simulated by the encoder in this setup. The frozen bits at the outer layer are generated as explained in \cite[Section IV]{honda12}. The frozen bits at the inner layer are chosen according to our shaper construction explained above. Note that the frozen bits at the inner and outer layer are forwarded to the decocer.

For the decoding we use the source coding decoder introduced in Section~\ref{sec:effQPC}. It is therefore clear that it is efficient, however it remains to be shown that it is reliable for channel coding which is done next.

\begin{mycor}
The encoder and decoder explained above have $O(N \log N)$ complexity.
\end{mycor}
\begin{IEEEproof}
This corollary is an immediate consequence of Lemma~\ref{lem:ApEncode} and Lemma~\ref{lem:ApDecoder}.
\end{IEEEproof}

\prlsection{Reliability}
To analyze the reliability of the decoder introduced above we start with a general lemma on the reliability of using the ``wrong'' compressor / decompressor pair in the problem of source coding.\footnote{Recall that we are using the decoder built for source coding, i.e.,\ for pairs of random variables $(X,Y)^N$, but we actually have only a (good) approximation of those.} 
\begin{mylem}
\label{lem:wrongdecoder}
Let $X$ and ${X}'$ be arbitrary random variables such that $\delta(X',{X})\leq \eta$ and let $\W$ denote an arbitrary stochastic map. 
 If $\C$ and $\D$ are a compressor / decompressor pair for $(X,\W(X))$, such that 
${\rm Pr}[\hat{X}\neq X]\leq \eta'$ where $\hat{X}=\D(\W(X),\C(X))$, then, for $\hat{X}'= \D(\W(X'),\C(X'))$,
\begin{align}
\Prv{\hat{X}'\neq X'}\leq \eta+\eta'.
\end{align}
\end{mylem}
\begin{IEEEproof}
  Note that the pairs $(X, \hat{X})$ and $(X', \hat{X}')$ are obtained
  from $X$ and $X'$ by applying the stochastic map that takes $x$ to
  $(x, \mathsf{D}(\mathsf{W}(x), \mathsf{C}(x)))$. Because the
  variational distance is non-increasing under such maps, we have
  \begin{align}
    \TD{\bigl(X, \hat{X}\bigr)}{\bigl(X', \hat{X}'\bigr)} \leq \TD{X}{X'} \leq  \eta \ .
  \end{align}
  Furthermore, defining $(X,X)$ to be the random variable
  $(X,\bar{X})$ with distribution $P_{X \bar{X}} = P_X \delta_{X
    \bar{X}}$, we have
  \begin{align}
    \TD{\bigl(X, X\bigr)}{\bigl(X,\hat{X}\bigr)} = \Pr[\hat{X} \neq X] \leq \eta' \ .
  \end{align}
  Hence, applying the triangle inequality, we obtain
  \begin{align}
    \TD{\bigl(X, X\bigr)}{\bigl(X', \hat{X}'\bigr)} \leq \eta + \eta' \ .
  \end{align}  
 Now note that the variational distance can also be written as
  \begin{align}
    \TD{A}{A'} = \sum_{a: \, P_{A}(a) \leq P_{A'}(a)} P_{A'}(a) - P_{A}(a) \ .
  \end{align}
  Applied to $A = (X, X)$ and $A'=(X', \hat{X}')$, and using that
  $P_{X X}(x, \hat{x}) = 0$ for $x \neq \hat{x}$, we immediately
  obtain
  \begin{align}
    \TD{\bigl(X, X\bigr)}{\bigl(X', \hat{X}'\bigr)} \geq \sum_{x \neq \hat{x}}
    P_{X'\hat{X}'}(x,\hat{x}) \ ,
  \end{align}
  which  implies that $\Pr[\hat{X}' \neq X'] \leq \eta + \eta'$. 
\end{IEEEproof}

Using Lemmas \ref{lem: newerrLem} and \ref{lem:wrongdecoder} the statement of reliability follows easily. 

\begin{mylem}
\label{lem:reliabilityCCC}
The error probability using the source coding decoder in the classical setup of channel coding as explained above is $P_{\rm err}=O(L\,2^{-M^\beta}+L2^{-\frac12L^{\beta'}})$ for any $\beta,\beta'<\frac12$. 
\end{mylem}

\begin{IEEEproof}
	For the polar source coding scheme, note that $\Prv{\mathcal{A}_i \cap \mathcal{B}_{i-1}^c} + x \in O(2^{-M^\beta})$, where $x$ is the probability that $\hat{V}_i \neq V_i$ given that a mistake previously occurred, but where we still give the correct $V^{i-1}$ to the decompressor. We can therefore upper bound $r$ in Lemma~\ref{lem: newerrLem} by $O(2^{-M^\beta})$~\cite{arikan10}.   Thus, the probability of incorrectly decoding any of the $K$ $V_i$ is $O(L 2^{-M^\beta})$; this is $\eta'$ in Lemma~\ref{lem:wrongdecoder}. Lemma~\ref{lem:inputoutput}  and the properties of $K$ give $\eta=O(L2^{-\frac12L^{\beta'}})$ for $\beta'<\tfrac{1}{2}$, establishing the theorem.
\end{IEEEproof}

\subsection{Decoding}
The decoding is equivalent to Bob's task in the entanglement distillation scheme. Note however that it has a slightly worse reliability since the coding scheme approximates the entanglement distillation setup by generating the frozen qubits at the inner and outer layer as explained above.
Recall that the inner and outer encoder forward the values of the frozen qubits to the decoder which is necessary to decode reliably and efficiently. An immediate corollary of Theorem~\ref{thm:Bob} states that
\begin{mycor}
For Pauli channels and the use of polar codes, the scheme introduced above has a $O(N \log N)$ complexity decoder. 
\end{mycor}


It remains to be shown that the efficient encoding and decoding introduced above is reliable. 
\begin{mythm}
The trace distance between the state produced by the decoder and the ideal, maximally entangled state is less than $\sqrt{2 \epsilon_2}+\sqrt{2M \epsilon_1}$ where $\epsilon_1=O(2^{-L^{\beta}})$, $\epsilon_2=O(L 2^{-M^{\beta'}}+L 2^{-\tfrac{1}{2}L^{\beta''}})$ for $\beta,\beta',\beta''<\tfrac{1}{2}$.
\end{mythm}
\begin{IEEEproof}
Similar as in the proof of Corollary~\ref{cor:quality} let $p_{\rm err}(\mathcal{D}_{A})$ denote the error probability of the amplitude-decoding operation $\mathcal{D}_A$ and $p_{\rm err}(\mathcal{D}_{P})$ the error probability of the phase-decoding process $\mathcal{D}_{P}$. The error probability for all decoding operations with respect to the amplitude basis is denoted by $p_{\rm err}(\mathcal{D}_{A^M})$. According to Proposition~\ref{prop:quality}, the trace distance between the scheme's output and a maximally entangled state of appropriate dimension is less than $\sqrt{2 p_{\rm err}(\mathcal{D}_{A^M}) } + \sqrt{2 p_{\rm err}(\mathcal{D}_P)}$. Using the union bound we obtain $p_{\rm err}(\mathcal{D}_{A^M}) \leq M p_{\rm err}(\mathcal{D}_{A})$. Since we use the standard polar decoder \cite{arikan09} for the amplitude error-correction, we have 
\begin{equation}
p_{\rm err}\left(\mathcal{D}_{A}\right)=O\left(2^{-L^{\beta}}\right) \quad \textnormal{for} \quad \beta<\frac{1}{2}.
\end{equation}
Furthermore, according to Lemma~\ref{lem:reliabilityCCC},
\begin{equation}
p_{\rm err}\left(\mathcal{D}_{P}\right)=O(L 2^{-M^{\beta'}}+L 2^{-\tfrac{1}{2}L^{\beta''}})
\end{equation}
for any $\beta',\beta''<\frac{1}{2}$, which proves the assertion.
\end{IEEEproof}


\section{Achieving Rates beyond the Coherent Information?}
\label{sec:beyond}
In Corollary~\ref{cor:CI} it is shown that the scheme achieves the coherent information. However if \eqref{eq:QuestionTight} were not tight the scheme could achieve a higher rate. We formulate a series of related open questions addressing this point and provide possible approaches to answer them.  
\renewcommand\themyquestion{1}
\begin{myquestion} \hypertarget{question}{}
Is it possible that $R > I(A\rangle B)_{\psi}$ for ${R>0}$?
\end{myquestion}
There are several ways to phrase the question above differently. Using \eqref{eq:bfquestion}, an equivalent formulation of Open Question~\hyperlink{question}{1} is the following:
\renewcommand\themyquestion{1$'$}
\begin{myquestion} \hypertarget{question1p}{}
Is it possible that $\lim \limits_{L\to \infty}\frac{1}{L}\Hc{X^{\bar A^{\setC}}}{B^L C^L X^{\bar A}}_{\Psi_3'} > 0$ for ${R>0}$?
\end{myquestion}

The following Proposition leads to a different, particularly clean reformulation of Open Questions~\hyperlink{question}{1} and \hyperlink{question1p}{1$'$}.
\begin{myprop} \label{prop:degraded}
The rate of the scheme introduced above can be written as
\begin{align}
R= - \Hc{A}{B}_{\psi} +\Hc{Z^A}{B}_{\psi}-\frac{1}{L}\Hc{Z^{\bar A^{\setC}}}{E^L}_{\Psi_3'}.
\end{align}
\end{myprop}
\begin{IEEEproof}
Recall that a possible rate expression is given in Theorem~\ref{thm:ratescheme}. Using the chain rule we can write
\begin{align}
&-\Hc{\bar A}{B^L C^L}_{\Psi_3'}  \nonumber \\
&= -\Hc{A^L}{B^L C^L}_{\Psi_3'}+\Hc{\bar A^{\setC}}{B^L C^L \bar A}_{\Psi_3'} \\
&= -L\Hc{A}{BC}_{\psi'} + \Hc{X^{\bar A^{\setC}}}{B^L C^L \bar A}_{\Psi_3'} - \left| \bar A^{\setC}\right| \label{eq:BertaUR} \\
&= -L\Hc{A}{BC}_{\psi'}-\Hc{Z^{\bar A^{\setC}}}{E^L}_{\Psi_3'} \label{eq:RenesUR}\\
&= -L \Hh{ABC}_{\psi'}-L\Hh{BC}_{\psi'}-\Hc{Z^{\bar A^{\setC}}}{E^L}_{\Psi_3'} \\
&= -L\Hc{AC}{B}_{\psi'}+L\Hc{Z^A}{B}_{\psi}-\Hc{Z^{\bar A^{\setC}}}{E^L}_{\Psi_3'}\\
&= -L\Hc{A}{B}_{\psi}+L \Hc{Z^A}{B}_{\psi}-\Hc{Z^{\bar A^{\setC}}}{E^L}_{\Psi_3'}.
\end{align}
Equality \eqref{eq:BertaUR} uses Lemma~\ref{lem:2pcertainty} and that $H(Z^{\bar A^{\setC}}|B^L C^L \bar A)=0$. Equation \eqref{eq:RenesUR} uses Lemma~\ref{lem:3pcertainty}. The remaining steps using the chain rule and the form of $\ket{\psi}$ and $\ket{\psi'}$. 
\end{IEEEproof}

Note that from the polarization phenomenon (cf.\ Theorem~\ref{thm:polarizationPhenomenon}) we know that $|\bar A^{\setC}| = L H(Z^A|B)_{\psi} - o(L)$ and hence we immediately see that we can bound the rate term as $R\geq I(A\rangle B)_{\psi}$. Proposition~\ref{prop:degraded} implies that we have an equivalent formulation of Open Question~\hyperlink{question}{1}, which might be easier to answer as it is a purely classical problem.\footnote{Note that Open Question~\hyperlink{question1pp}{1$''$} is formulated in a purely classical framework in \cite[Section V.A]{sutter13}.}
\renewcommand\themyquestion{1$''$}
\begin{myquestion} \hypertarget{question1pp}{}
Is it possible that $\lim \limits_{L \to \infty} \frac{1}{L}H(Z^{\bar A^{\setC}}|E^L)_{\Psi_3'} < \lim \limits_{L\to \infty} \frac{1}{L}| \bar A^{\setC}|$ for $R>0$?
\end{myquestion}

Equation~\eqref{eq:lessnoisyrates} states that for \emph{less noisy} channels the coherent information is optimal. Therefore we must be able to show that $R=I(A\rangle B)_{\psi}$.  Using the less noisy characterization (cf.\ Section~\ref{sec:def}), we can write $H(Z^{\bar A^{\setC}}|E^L)_{\Psi_3'} \geq H(Z^{\bar A^{\setC}}|B^L)_{\Psi_3'} = L H(Z^A|B)_{\psi}+o(L)$, where the last step follows from the polarization phenomenon as stated in Theorem~\ref{thm:polarizationPhenomenon}. Using \eqref{eq:RenesUR}, we obtain for sufficiently large $L$, that $R\leq I(A\rangle B)_{\psi}$ which together with Corollary~\ref{cor:CI} proves that for less noisy channels the rate of our scheme is equal to the coherent information.

Recent advances in understanding the polarization phenomenon~\cite{RSH15} may be useful to resolve this open question, since Questions~\hyperlink{question1p}{1$'$} and \hyperlink{question1pp}{1$''$} when applied to quantum polar codes involve a statement about the structure of the polarized sets (see Section~\ref{sec_polpheno}). 
A possible indication why our scheme could indeed achieve rates beyond the coherent information is the observation that the outer layer can introduce degeneracies for the code at the inner layer. In addition, the states at the inner layer are not product states, which would rule out having a rate beyond the coherent information.\footnote{Note that for product states the von Neumann entropy is additive.}

\section{Achieving the Quantum Capacity?} \label{sec:achievingQ}
In this section we state a second open problem which addresses the question of whether it is possible that our scheme achieves the quantum capacity for channels where the coherent information is not optimal. This question is related to the one introduced in Section~\ref{sec:beyond} before. Whereas Open Questions~\hyperlink{question}{1}, \hyperlink{question1p}{1$'$} and \hyperlink{question1pp}{1$''$} ask if it is possible to achieve a rate beyond the coherent information for certain quantum channels, Open Question~\hyperlink{question:LN}{2} in this section raises the more specific question if for certain quantum channels our scheme achieves the quantum capacity.\footnote{The existence of a channel fullfilling the condition in Question~\hyperlink{question:LN}{2} which is not less noisy  would imply a positive answer to Questions~\hyperlink{question}{1}, \hyperlink{question1p}{1$'$} and \hyperlink{question1pp}{1$''$}.}
As depicted in Figure~\ref{fig:channelview}, we can define a super-channel $\mathcal{N}':\mathfrak{T}(\mathcal{H}_{\bar A}) \to \mathfrak{T}(\mathcal{H}_{B^{\otimes L}} \otimes \mathcal{H}_{B_C})$ which consists of an inner encoding block and $L$ basic channels $\mathcal{N}$. Then we have
\begin{myprop} \label{prop:LN}
For sufficiently large $L$, the channel $\mathcal{N}'$ is approximately less noisy, irrespective of $\mathcal{N}$.
\end{myprop}
\begin{IEEEproof}
Recall that as explained in Section~\ref{ssec:EDscheme}, we have $p_{\mathrm{err}}(Z^{\bar A}|B^L B_C)_{\Psi_2}\leq \epsilon_1$. Using Fano's inequality we obtain
\begin{equation}
\Hc{Z^{\bar A}}{B^L B_C}_{\Psi_2} \leq \Hb(\epsilon_1) + \epsilon_1 \log \dim \bar A =: \xi.
\end{equation}
Recall that we choose the CSS code used in our scheme such that $\xi \to 0$ for $L \to \infty$. For example using quantum polar codes, we have $\epsilon_1=O\left(2^{-L^{\beta}} \right)$ for $\beta<\tfrac{1}{2}$ and $\log \dim \bar A = O(L)$. The assertion then follows from \cite[Lemma 12]{sutter13}.
\end{IEEEproof}
Proposition~\ref{prop:LN} and \eqref{eq:lessnoisyrates} imply that
\begin{align}
&\lim \limits_{L \to \infty} \frac{1}{L} Q(\mathcal{N}') \nonumber \\
&\hspace{3mm}= \lim \limits_{L \to \infty} \frac{1}{L} Q_1(\mathcal{N}') \\
&\hspace{3mm}= \lim \limits_{L \to \infty} \frac{1}{L} \max \limits_{\phi^{\bar A \bar A'}} \CI{\bar A}{B^L C^L}_{\Psi_3,\,\, \mathcal{N'}^{\bar A'\to B^L C^L}(\phi^{\bar A \bar A'})}  \\
&\hspace{3mm}= R,
\end{align}
where the last step follows from Theorem~\ref{thm:ratescheme}. Our scheme hence achieves the capacity of the super-channel. This raises the question of when the super-channel $\mathcal{N}'$ has the same capacity than the original channel $\mathcal{N}$, i.e.,\ how much is lost in the first layer of our scheme.
\renewcommand\themyquestion{2}
\begin{myquestion} \hypertarget{question:LN}{}
Under which conditions does $\tfrac{1}{L} Q(\mathcal{N}')=Q(\mathcal{N})$ hold?
\end{myquestion}
Equation~\ref{eq:lessnoisyrates} and Theorem \ref{thm:ratescheme} imply that $\tfrac{1}{L} Q(\mathcal{N}')=Q(\mathcal{N})$ holds if $\mathcal{N}$ is less noisy. However, it is not yet known if this is necessary for the condition to hold. 


\section{Secret Key Distillation \& Private Channel Coding}  \label{sec:skd}
\subsection{Scheme, Reliability, Secrecy and Rate}
With minor changes, the above entanglement distillation scheme also works for secret key distillation in a setup where we have two quantum parties Alice and Bob as well as a quantum adversary Eve. Consider the scenario in which Alice and Bob share an additional ``shield'' system $S$ \cite{horodecki_secure_2005,renes3short}. A shield is any system not held by the eavesdropper Eve but  nevertheless cannot be used for amplitude IR by Alice and Bob, where the amplitude information is used to create the secret key. One can show that 
\begin{align}
&p_{\rm err}\left(X^{\hat A}|B^NC^N B_D S\right)\leq \epsilon_2\quad\textnormal{and} \label{eq:hold1}\\
&p_{\rm err}\left(Z^{\hat A}|B^N B_C^M\right)\leq M\epsilon_1, \label{eq:hold2}
\end{align}
 characterizes a state where $Z^{\hat A}$ can be used as a secret key. 
 
Due to the uncertainty principle, the secrecy of the amplitude information from Eve is ensured if Alice and Bob could implement phase IR. Moreover, it is only necessary that phase IR \emph{could} be performed; it is not necessary to actually do it, and thus $S$ may be used as side information in the procedure, no matter who holds which parts of $S$ \cite{renes10_2}. 

An observable $Z^A$ is approximately secure if the trace distance to the ideal case is small. Therefore we introduce for $\psi^S=\Trs_{A}[\psi^{AS}]$, 
\begin{equation}
p_{\rm{secure}}\left(Z^A|S\right):=\tfrac{1}{2}\norm{\psi^{AS}-\tfrac{1}{d}\mathbbm{1}\otimes \psi^S }_1,
\end{equation}
where $\norm{M}_1:=\Trs[{\sqrt{M^{\dagger}M}}]$. Note that $p_{\rm{secure}}(Z^A|S)_{\psi}\leq \epsilon$ means that the actual $\epsilon$-secure $Z^A$ can only be distinguished from the ideal, a perfect key, with probability at most $\epsilon$. 

\begin{mycor} \label{cor:skdReliability}
For the scheme explained above we have
\begin{align}
&p_{\rm secure}\left(X^{\hat A}\!\left|E^N E_C^M S \right. \right) \leq \sqrt{2 \epsilon_2} \quad \textnormal{and}\\
&p_{\rm secure}\left(Z^{\hat A}\! \left|E^N E_C^M E_D \right.\right) \leq \sqrt{2 M \epsilon_1}.
\end{align}
\end{mycor}

\begin{IEEEproof}
We show that \eqref{eq:hold1} and \eqref{eq:hold2} are satisfied. The assertion then immediately follows from \cite[Theorem 4.1.]{renes_habi,renes10_2}.
Due to the i.i.d.\ structure of the scheme and the union bound, it is sufficient to show that $p_{\rm err}(Z^{\bar A}|B^L B_C)\leq \epsilon_1$ in order to prove \eqref{eq:hold2}. This however is a direct consequence of Bob's decoding task for the entanglement distillation setup, as introduced in Section~\ref{ssec:EDscheme}.

Inequality \eqref{eq:hold1} is equivalent to the error probability for the phase decoding Bob performs (as explained in Section~\ref{ssec:EDscheme}) with the difference that there is extra side information $S$ in this case. Fortunately it is straightforward to modify the decoder to the setup with extra side information \cite{arikan10}.
\end{IEEEproof}

The rate for this scenario can be computed analogously as in Section~\ref{ssec:EDscheme}, which leads to
\begin{align}
R&\geq 1-\Hc{Z^A}{B}_{\psi}-\Hc{X^A}{BCS}_{\psi} \label{eq:ratee1}\\
 &=\Hc{Z^A}{E}_{\psi}-\Hc{Z^A}{B}_{\psi} \label{eq:ratee2}.
\end{align}
The final step uses the exact-uncertainty relation given in \cite{renes2short}, which ensures that $H(Z^A|E)_{\psi}+H(X^A|BCS)_{\psi}=1$. Note that we no longer obtain the coherent information $-H(A|B)_{\psi}$, since $E$ no longer purifies $AB$.
The reliability of the secret key distillation scheme is analogous to that of the entanglement distillation scheme (cf. Proposition~\ref{prop:quality}).

\subsection{Using Quantum Polar Codes for Pauli or Erasure Channels}
As mentioned above the secret key distillation scheme is very similar to the entanglement distillation scheme introduced in Section~\ref{sec:ed}. More precisely, Alice's first task, i.e.,\ the $M$ amplitude information reconciliation blocks are identical as in the entanglement distillation. The phase IR step is slightly different as one has to consider side information $S$, which leads to a different set of frozen qubits. Furthermore, Alice does not send the outcomes from measuring $ \hat A^{\setC}$ to Bob, but simply keeps them secret from Eve.  Alice's task thus can be done with $O(N \log N)$ complexity as proven in Theorems~\ref{thm:alice} and \ref{thm:Bob}.

Bob's task is also similar to the decoding he performs in the entanglement distillation setup. He first decodes the amplitude $Z^{A^N}$, which can be done with a standard classical polar decoder. He next computes the value of $Z^{\hat A}$ using the details of the phase IR code. Hence Bob's decoding operation has $O(N \log N)$ complexity.

The reliability and secrecy are as given in Proposition~\ref{prop:quality} and Corollary~\ref{cor:skdReliability} with $\epsilon_1 =O(2^{-L^{\beta'}})$ and $\epsilon_2=O(L 2^{-M^{\beta}})$ for any $\beta,\beta' < \tfrac{1}{2}$.

The secret key distillation scheme described above also works in a purely classical setup, since the phase IR protocol can be turned into a privacy amplification protocol needing only classical operations as shown in \cite{sutter13}.
Our scheme for classical one-way secret key agreement improves practically efficient protocols where the eavesdropper has no prior knowledge \cite{abbeITW} or/and degradability assumptions are required \cite{chou13}.

\subsection{Private Communication}
The quantum coding scheme can be used for efficient private channel coding at a high rate (as given in \eqref{eq:ratee1} and \eqref{eq:ratee2}). As in Section~\ref{sec:channelcoding}, the idea is to run key distillation in reverse, simulating the measurement outputs with appropriately-chosen random inputs. These are then the frozen bits. The frozen bits of the inner and outer encoder can be sent over an insecure channel to Bob, since privacy is ensured by the outer layer whose frozen bits are uncorrelated to the message bits since the corresponding protocol produces entanglement.

In \cite{sutter13}, it is shown how to use a similar scheme in a purely classical scenario for private channel coding at the secrecy capacity such that there exists an encoding and decoding operation that have an essentially linear complexity. Our scheme improves previous work on efficient wiretap coding, where either only weak secrecy could be proven \cite{vardy11} or/and degradability assumptions are required \cite{sasoglu13}.

\section{Conclusion}
We have constructed a protocol that can be used to perform reliable entanglement distillation or quantum communication at a rate equal to (or possibly larger than) the coherent information. Compared to previous work in this area, our scheme does not require any preshared entanglement and achieves the coherent information also for asymmetric channels (where $\phi^{A A'}$ in \eqref{eq:cohinf} is not necessarily a Bell state). When communicating over a Pauli or erasure channel using polar codes, encoding and decoding can be performed with a number of operations essentially linear in the blocklength. We have also shown how the protocol can be modified for efficient, high-rate secret key distillation and private channel coding.




\bibliographystyle{IEEEtran}
\bibliography{../bibtex/header,../bibtex/bibliofile}

\end{document}